\pdfoutput=1
\documentclass[12pt,a4paper]{amsart}
\usepackage[a4paper,inner=2.5cm,outer=2.5cm,top=2.5cm,bottom=2.5cm]{geometry}
\usepackage{amsmath,amssymb,amsthm,enumerate,mathtools,stmaryrd
}
\usepackage{hyperref}
\hypersetup{colorlinks=true,linkcolor=blue,citecolor=teal,filecolor=magenta,urlcolor=cyan}

\usepackage{longtable}

\usepackage{makecell}
\usepackage{graphicx}

\usepackage{float}

\usepackage{extarrows}

\usepackage{xcolor}

\def\xvital{\textsf{key}}
\def\yvital{$\vee$-\textsf{key}}
\def\regular{{non-special}}
\def\nonregular{{special}}

\usepackage{url}

\usepackage[backend=bibtex,style=alphabetic,sorting=nyt,isbn=false,url=false,doi=true,maxalphanames=10,minalphanames=4,mincitenames=4,maxcitenames=10,minnames=4,maxnames=10,giveninits=true,maxbibnames=99]{biblatex}
\theoremstyle{plain}
\newtheorem{theorem}{Theorem}[section]
\newtheorem{proposition}[theorem]{Proposition}

\newtheorem{lemma}[theorem]{Lemma}

\theoremstyle{definition}
\newtheorem{definition}[theorem]{Definition}
\newtheorem{example}[theorem]{Example}

\makeatletter
\@addtoreset{proofpart}{theorem}
\makeatother
\theoremstyle{remark}
\newtheorem{remark}[theorem]{Remark}

\newcommand{\Z}{\mathbb{Z}}
\newcommand{\C}{\mathbb{C}}
\newcommand{\bW}{\mathbb{W}}
\newcommand{\bK}{\mathbb{K}}

\newcommand{\cS}{\mathcal{S}}

\newcommand{\cP}{\mathcal{P}}

\newcommand{\cW}{\mathcal{W}}
\newcommand{\cT}{\mathcal{T}}

\newcommand{\res}{\mathop{\rm res}}

\newcommand{\restr}[2]{\mathop{\big\lfloor_{{#1}\to {#2}}}}
\newcommand{\set}[1]{\llbracket {#1} \rrbracket}

\newcommand{\lcoord}{\xi}

\title%[Degenerate and irregular TR]
{Degenerate and irregular topological recursion}

\author[A.~Alexandrov]{A.~Alexandrov}
\address{A.~A.: Center for Geometry and Physics, Institute for Basic Science (IBS), Pohang 37673, Korea
}
\email{alex@ibs.re.kr}

\author[B.~Bychkov]{B.~Bychkov}
\address{B.~B.: Department of Mathematics, University of Haifa, Mount Carmel, 3498838, Haifa, Israel}
\email{bbychkov@hse.ru}

\author[P.~Dunin-Barkowski]{P.~Dunin-Barkowski}
\address{P.~D.-B.: Faculty of Mathematics, HSE University, Usacheva 6, 119048 Moscow, Russia; HSE--Skoltech International Laboratory of Representation Theory and Mathematical Physics, Skoltech, Bolshoy Boulevard 30 bld. 1, 121205 Moscow, Russia; and NRC “Kurchatov Institute” -- ITEP, 117218 Moscow, Russia}
\email{ptdunin@hse.ru}

\author[M.~Kazarian]{M.~Kazarian}
\address{M.~K.: Faculty of Mathematics, HSE University, Usacheva 6, 119048 Moscow, Russia; and Igor Krichever Center for Advanced Studies, Skoltech, Bolshoy Boulevard 30 bld. 1, 121205 Moscow, Russia}
\email{kazarian@mccme.ru}

\author[S.~Shadrin]{S.~Shadrin}
\address{S.~S.: Korteweg-de Vries Institute for Mathematics, University of Amsterdam, Postbus 94248, 1090GE Amsterdam, The Netherlands}
\email{S.Shadrin@uva.nl}

\begin{document}

\begin{abstract} We use the theory of $x-y$ duality to propose a new definition / construction for the correlation differentials of topological recursion; we call it \emph{generalized topological recursion}. This new definition coincides with the original topological recursion of Chekhov--Eynard--Orantin in the regular case and allows, in particular, to get meaningful answers in a variety of irregular and degenerate situations. \end{abstract}

\maketitle

\tableofcontents

\section{Introduction} 

\subsection{Various kinds of topological recursion} The theory of topological recursion \cite{CEO,EO-1st} emerges from the theory of matrix models and is originally due to Chekhov, Eynard, and Orantin (CEO). Its main construction is a recursive procedure that  produces a system of symmetric meromorphic $n$-differentials $\omega^{(g)}_n$ on $\Sigma^n$, $g\geq 0$, $n\geq 1$, $2g-2+n>0$, where $\Sigma$ is a Riemann surface called the spectral curve (by default, we have no assumptions on $\Sigma$, in particular throughout the text it does not have to be compact unless it is explicitly demanded in some statements). The initial data includes two meromorphic functions, $x$ and $y$, with quite different roles, and a meromorphic bi-differential $B$ on $\Sigma^2$ with the only pole along the diagonal with bi-residue $1$. These input data can be either further specified or generalized depending on the problem. 

The recursion itself is performed via certain local computations near the critical points of $x$; we recall the precise definition and its variations below. Here we just make a small survey of the typical local setups that one has to consider in various applications:
\begin{itemize}
	\item The critical points of $x$ (let us denote them by $q_1,\dots,q_N$), are simple; $y$ is regular at $q_i$ and $dy|_{q_i}\not=0$, $i=1,\dots,N$. This is the standard setup of the CEO recursion~\cite{EO-1st}. 
	\item In the standard setup one might relax the requirement that the critical points of $x$ are simple; they might be of arbitrary multiplicity, but the requirements for $dy$ are still the same. In this case, sometimes called {\em degenerate}, a definition was proposed by Bouchard and Eynard in~\cite{BE}. 
	\item In the standard setup one can relax the condition that $x$ and $y$ are meromorphic functions and assume only that $dx$ and $dy$ are meromorphic; this leads to the definition of \emph{logarithmic topological recursion}~\cite{ABDKS-logTR-xy}.
	\item In the standard setup one might relax the requirement that $y$ is regular at $q_i$. For instance, one might have a simple pole at $q_i$, $i=1,\dots,N$. This case, sometimes called {\em irregular}, was studied by Chekhov and Norbury~\cite{CN}.
	\item One can combine these two ideas of generalization, that is, both allow $x$ to have critical points $q_1,\dots,q_N$ of arbitrary multiplicity and $y$ have arbitrary meromorphic behavior at $q_i$, $i=1,\dots,N$. In this case, the suitable application of the Bouchard--Eynard formula might produce non-symmetric differentials, so some further conditions on local compatibility of $x$ and $y$ are required. This case was studied by Borot et al. in~\cite{BBCCN}. 
\end{itemize}
These setups are related to each other via natural limiting procedures, and the compatibility with the limits was studied in~\cite{limits}.

A particular problem that prompts us to look for a revision of the definition of topological recursion in the most general case is the fact that in the most general setup the standard formulas do not work in the sense that they do not produce symmetric differentials. In addition to these generalizations, as mentioned above, the concept of logarithmic topological recursion was proposed in~\cite{ABDKS-logTR-xy}, which captures the natural effects of $y$ having logarithmic singularities~\cite{hock2023xy}, and the goal is to have a universally working definition that would incorporate this generalization as well. 

\subsection{Towards a new definition through \texorpdfstring{$x-y$}{x-y} duality}

One of the recent developments (prompted in particular by the two-matrix models, by a huge amount of problems in enumerative geometry and  combinatorics universally resolved by topological recursion, as well as by applications in the free probability theory) is a thorough study of the behavior of the differentials $\omega^{(g)}_n$ under the so-called $x-y$ duality~\cite{EO-xy,borot2023functional,hock2022xy,hock2022simple,ABDKS1,hock2023laplace} and more general symplectic duality~\cite{BDKS2,bychkov2023symplecticdualitytopologicalrecursion,alexandrov2023topologicalrecursionsymplecticduality,ABDKS-log-sympl}. In a nutshell, the $x-y$ duality is a procedure that exchanges the roles of the functions $x$ and $y$ in the setup of topological recursion, and its effect can be captured by explicit closed formulas that also have natural interpretation in KP integrability~\cite{BDKS1,BDKS3,ABDKS3}; and symplectic duality is its natural generalization. Some earlier reincarnation of the $x-y$ duality in a special case is the theory of $p-q$ duality in string theory~\cite{Kharchev-Marshakov}.

An idea noted initially in~\cite{ABDKS1} suggests that the compatibility with the $x-y$ formulas and further control of emerging singularities can serve as a source of a new interpretation (in the regular cases) or a new definition (in degenerate and irregular cases) of topological recursion. So far this idea was tested in applications where $x$ and $y$ are allowed to have logarithmic singularities~\cite{hock2023xy,ABDKS-logTR-xy,ABDKS-log-sympl}. This led to the development of the so-called logarithmic topological recursion.

The purpose of the present paper is to explore this idea further and to revisit the definition of topological recursion using the compatibility with the $x-y$ duality basically as an axiom. We give a new definition, discuss its identification with the CEO topological recursion in the standard case and provide a comparison with the Bouchard--Eynard recursion in degenerate / irregular cases. 

\subsection{Advantages and properties of the new approach}
From the point of view of assumptions one has to impose on the initial data, the new approach is much more flexible than any other approach available so far: we do not have any assumptions on the behavior of $x$ and $y$. They might have overlapping poles and critical points, possible logarithmic singularities (that is, we only use $dx$ and $dy$ and they can be absolutely arbitrary meromorphic $1$-forms); the assumptions for them in the definition are entirely symmetric.  

The new definition is also very well compatible with the deformations of the initial data. It behaves nicely in families and provides concrete tools to trace what happens in certain degenerate limits. Furthermore, the new definition provides a greater degree of flexibility, enabling the discovery of novel natural solutions that were previously overlooked by traditional approaches. These solutions possess all the desirable properties typically associated with topological recursion.

The new definition also immediately has quite a few applications. First and foremost, it allows to uniformly extend a variety of results obtained in the previous papers of the authors under the technical assumption of being in general position. Also, it provides a proper context to a recent interpretation of the intersection theory of the so-called $\Theta^{(r,s)}$-classes obtained in~\cite{CGS}.

%And the last but not the least: if the spectral curve has genus $0$, the differentials constructed by the new definition are always KP integrable in the sense of~\cite{ABDKS3}, which generalizes the result obtained in~\cite{alexandrov2024topologicalrecursionrationalspectral}.

And last but not least: if the spectral curve has genus $0$, the differentials constructed by the new definition are always KP integrable in the sense of~\cite{ABDKS3}. This generalizes the result obtained in~\cite{alexandrov2024topologicalrecursionrationalspectral}. 

\subsection{Organization of the paper} 

Section~\ref{sec:GenTR} contains the main construction of the paper. We first recall there the precise definition of the original topological recursion in the standard local setup, and then we immediately introduce the new definition which we call the \emph{generalized topological recursion}, along with the basic statements of its properties. The rest of the paper is devoted to the analysis of this new definition. 

In Sections~\ref{sec:xyswap} and~\ref{sec:loop} we prove the two main basic properties of the new definition, namely, that it produces symmetric differentials and that it reproduces the original CEO topological recursion for a particular choice of the input data. To this end we need to develop some extra toolkit that is strongly tied to KP integrability~\cite{ABDKS1}: the formal properties of the $x-y$ swap formulas and the theory of the loop equations, both of which are introduced as abstract concepts (that is, not in direct connection to the definition of the generalized topological recursion) in the respective sections. More precisely,

\begin{itemize}
	\item In Section~\ref{sec:xyswap} we recall the basic setup of the $x-y$ duality as an abstract system of relations between two sets of multi-differentials. We explain the connection of these relations to the setup of generalized topological recursion and prove that the differentials of generalized topological recursion are symmetric in all their variables. We also prove that in the algebraic case (that is, under the assumption that the spectral curve is compact), the $x-y$ duality is compatible with the generalized topological recursion.
	\item In Section~\ref{sec:loop} we recall the approach to the loop equations for the systems of multi-differentials developed in~\cite{ABDKS1}. The loop equations allow to establish the identification of the original CEO topological recursion and the generalized topological recursion, under the standard assumptions of the original CEO topological recursion.
\end{itemize}

In Section~\ref{sec:compatibility} we discuss the compatibility (and eventual differences) of the generalized topological recursion with different versions of topological recursion mentioned above. In particular, we explain why the generalized topological recursion incorporates as a special case the so-called logarithmic topological recursion, developed in a bunch of recent papers as a natural extension of the original CEO topological recursion. We also prove that the generalized topological recursion agrees with the setup of Chekhov--Norbury and explain in which cases it agrees with the Bouchard--Eynard version of topological recursion. We also give examples and explain the origin of the disagreement of the generalized topological recursion with the Bouchard--Eynard topological in the most general case.

In Section~\ref{sec:KP} we prove the KP integrability for the differentials of the generalized topological recursion in the case of the rational spectral curve.

Finally, in Section~\ref{sec:rs-examples} we provide some explicit formulas and computations of expansion in the basic cases $x=z^r$, $y=z^s$, which partly have known enumerative meaning due to the results of~\cite{EO-1st,Norbury,chidambaram2023relationsoverlinemathcalmgnnegativerspin,CGS}, as well as $x=z^r-\log z$, $y=z^s$~\cite{LPSZ}. 

\subsection{Notation} Throughout the text we use the following notation, by now the standard one in the papers on topological recursion and its applications:
\begin{itemize}
	\item $\set{n}$ denotes the set $\{1,\dots,n\}$.
	\item $z_I$ denotes $\{z_i\}_{i\in I}$ for any subset $I\subseteq \set{n}$.
	\item $[u^d]$ denotes the operator that extracts the corresponding coefficient from the whole expression to the right of it, that is, $[u^d]\sum_{i=-\infty}^\infty a_iu^i \coloneqq a_d$.
	\item
	$\restr{u}{v}$ denotes the operator of substitution (or restriction) applied to the whole expression to the right of it, that is, $\restr{u}{v} f(u) \coloneqq f(v)$.
	\item $\cS(u)$ denotes the formal power series $u^{-1}(e^{u/2} - e^{-u/2}) = \sum_{i=0}^\infty \frac{u^{2i}}{2^{2i}(2i+1)!}$;
	\item $d\frac{1}{dx}$ denotes an operator that acts on a $1$-differential $\omega$ as $d(\frac{\omega}{dx})$.
\end{itemize}

\subsection{Acknowledgments} A.~A. was supported by the Institute for Basic Science (IBS-R003-D1). A.~A. is grateful to MPIM in Bonn for hospitality and financial support.  Research of B.~B. was supported by the ISF grant 876/20. P.D.-B. and M.K. worked on this project within the framework of the Basic Research Program at HSE University. S.~S. was supported by the Dutch Research Council grant OCENW.M.21.233.

We would like to thank the referees for the very careful reading of the paper and many useful remarks. In particular, one of the referees suggested that the interplay of the generalized topological recursion and KP integrability can help to revisit the rich system of interrelations between topological recursion and further topics in integrability, such as the theory of isomonodromic deformations.

\section{Generalized topological recursion} \label{sec:GenTR}

\subsection{Original CEO topological recursion} \label{sec:nondegenerateTR}

Let $(\Sigma,x,y,B)$ be the spectral curve data. In one of the most basic settings, there are no assumptions on $\Sigma$ and it can be, for instance, a collection of small discs, both $x$ and $y$ are assumed being meromorphic functions on $\Sigma$, all zeroes of $dx$ are simple, and $dy$ is holomorphic and nonvanishing at zeroes of~$dx$, and $B$ is  a meromorphic bi-differential on $\Sigma^2$ with the only pole along the diagonal with bi-residue $1$. In this setting, the original CEO topological recursion~\cite{CEO,EO-1st} defines the $n$-differentials $\omega^{(g)}_n$ in the unstable cases by
\begin{equation}
\omega^{(0)}_1(z)=y(z)\;dx(z),\qquad \omega^{(0)}_2(z_1,z_2)=B(z_1,z_2),
\end{equation}
and for $2g-2+n\ge0$ recursively by a two-step procedure that we describe below.

Let  $\{q_1,\dots,q_N\}$ be the set of zeros of $dx$, and let $\sigma_i$ be the deck transformation of $x$ near $q_i$, $i=1,\dots,N$. At the first step we define the germs of differentials  $\bar\omega^{(g)}_{n+1 \mathop{|} i}$,
\begin{align}
\label{eq:CEO}
\bar\omega^{(g)}_{n+1 \mathop{|} i}(z,z_{\set{n}})  &= \tfrac{1}{(y(z)-y(\sigma_i(z))dx(z)}
\biggl( \omega^{(g-1)}_{n+2,0} (z,\sigma_i(z),z_{\set{n}})
\\\notag&\qquad\qquad	+ {\displaystyle\sum_{\substack{g_1+g_2=g,~ I_1\sqcup I_2 =\set{n}\\ (g_i,|I_i|)\not= (0,0)}}} \omega^{(g_1)}_{|I_1|+1,0}(z,z_{I_1})\omega^{(g_2)}_{|I_2|+1,0}(\sigma_i(z),z_{I_2})\biggr),
\end{align}
in a vicinity of the points $q_1,\dots,q_N$. At the second step, we set
\begin{align}
\label{eq:proj}
\omega^{(g)}_{n+1}(z,z_{\set{n}})&=\sum_{i=1}^N\res\limits_{\tilde z=q_i}
\bar\omega^{(g)}_{n+1\mathop{|} i}(\tilde z,z_{\set{n}})\int\limits^{\tilde z}B(\cdot,z),
\end{align}
where the integration constants in the last formula can be chosen arbitrarily since $\bar\omega^{(g)}_{n+1\mathop{|} i}$ has trivial residue at $q_i$, $i=1,\dots,N$.

The meaning of these relations is as follows. The differential $\bar\omega^{(g)}_{n+1}\coloneqq \sum_{i=1}^N \bar\omega^{(g)}_{n+1\mathop{|} i}$, which is a combination of the germs \eqref{eq:CEO}, is considered as a ``preliminary approximation'' of $\omega^{(g)}_{n+1}$. It is defined only for $z$ in a vicinity of the union of points $q_1,\dots,q_N$ and we have
\begin{equation}
\omega^{(g)}_{n+1}(z,z_{\set{n}})=\bar \omega^{(g)}_{n+1}(z,z_{\set{n}})+\text{(holomorphic in $z$)},\quad z\to q_j,\quad j=1,\dots,N.
\end{equation}
Respectively, the integral transformation~\eqref{eq:proj} recovers $\omega^{(g)}_{n+1}$ as a $1$-form in $z$, which is global meromorphic, has no poles beyond $\{q_1,\dots,q_N\}$, the principal parts of its poles at the points~$q_j$ coincide with those of $\bar \omega^{(g)}_{n+1}$. In general, such differentials are not unique, and inherit the properties of the bi-differential $B$. For instance, if $\Sigma$ is compact, one can assume that $B$ is uniquely defined by the requirement that it has trivial $\mathfrak{A}$-periods, and then so does $\omega^{(g)}_{n+1}$. 

\bigskip

In the framework of this definition, one of the ways to think of the main construction 
in this paper is the following. We provide an alternative expression for $\bar \omega^{(g)}_{n+1}$ which differs from~\eqref{eq:CEO} by a holomorphic summand and thus produces the same differentials~$\omega^{(g)}_{n+1}$ in the above setting. 
Moreover, our substitute expression for $\bar \omega^{(g)}_{n+1}$ is purely algebraic in a sense that it does not involve the topology of the ramification defined by the function~$x$ and does not even involve any information on the orders of zeroes of~$dx$. 

This leads to a form of topological recursion that makes sense in a much more general setting: as we shall see, it will be sufficient just to assume that $dx$ and $dy$ are arbitrary meromorphic differentials on~$\Sigma$ with possibly nonzero residues and periods and with no conditions imposed on the orders and positions of their zeroes and poles.

\subsection{Initial data of generalized topological recursion}

Let $dx,dy$ be two meromorphic differentials on a smooth complex curve~$\Sigma$.

\begin{definition}\label{def:singpoint}
For a given point $q\in\Sigma$ and a local coordinate $z$ at this point consider the local expansions of $dx$ and $dy$:
\begin{align}
dx=a\,z^{r-1}(1+O(z))dz,\quad dy=b\,z^{s-1}(1+O(z))dz,\quad a,b\ne0,\ r,s\in\Z.	
\end{align}
The point $q$ is called \regular{} if either $r=s=1$ or $r+s\le0$, and \nonregular{} otherwise.
\end{definition}

\begin{definition} The \emph{initial data of generalized topological recursion} is a tuple $(\Sigma,dx,dy,B,\cP)$,
where
\begin{itemize}
\item $\Sigma$ is a smooth complex curve.
\item $B$ is a symmetric bi-differential on $\Sigma^2$ with a second order pole on the diagonal with bi-residue~$1$ and no other poles. 
\item $dx$, $dy$ are two arbitrary nonzero meromorphic differentials on~$\Sigma$.
\item $\cP=\{q_1,\dots,q_N\}$ is an arbitrarily chosen finite subset in the set of \nonregular{} points (\nonregular{} in the sense of Definition~\ref{def:singpoint}).
\end{itemize}
\end{definition}

\begin{definition}
	The points of~$\cP$ are called the \xvital{} points.
	The complementary to $\cP$ subset in the set of \nonregular{} points is denoted by~$\cP^\vee$ and its points are called 
	\yvital.
\end{definition}

\begin{definition}
The \emph{$x-y$ dual initial data} of generalized topological recursion is formed by the tuple $(\Sigma,dy,dx,B,\cP^\vee)$. It is assumed that $\cP^\vee$ is also finite. 
\end{definition}

\begin{remark} If $\Sigma$ is compact, then $B$ is defined 
	%up to a biholomorphic summand which can be fixed 
	by the condition of vanishing $\mathfrak A$-periods for some chosen basis of  $(\mathfrak{A},\mathfrak{B})$-cycles on $\Sigma$.
\end{remark}

\begin{remark} Note that we assume no restriction on the orders and positions of poles and zeroes, residues and periods of $dx$ and $dy$.
\end{remark}

\begin{remark}
Note that the notions of \regular{} and \nonregular{} points are symmetric with respect to the swap of the roles of $dx$ and $dy$. In order to set the initial data of recursion and its dual we need to choose for each \nonregular{} point in an arbitrary way whether it is 
\xvital{} or \yvital. 

If $dx$ and $dy$ have no common zeroes and no simple poles, then there is a canonical choice: we can define a \nonregular{} point being 
\xvital{} or \yvital{} if it is a zero of $dx$ or $dy$, respectively. However, at a common zero of $dx$ and $dy$ there is no preferable choice. Besides,  in certain cases it might be useful to consider some zeroes of $dx$ as \yvital{} 
and some zeroes of $dy$ as \xvital{} 
\nonregular{} points.
\end{remark}

\subsection{Definition of generalized topological recursion}

For a given collection of differentials $\{\omega^{(g)}_n\}_{g\geq 0, n\geq 1}$ such that $\omega^{(0)}_1=y\,dx$ we set
\begin{equation}
\omega_n=\sum_{\substack{g\ge0\\(g,n)\ne(0,1)}} \hbar^{2g-2+n}\omega^{(g)}_{n}.
\end{equation}
Assuming that $\omega_n$, $n\geq 1$, are globally defined meromorphic $n$-differentials,  we define a new collection of differentials
$\cW_n(z,u;z_1,\dots, z_n)=\sum_{g=0}^\infty\hbar^{2g-1+n}\cW^{(g)}_n$ on $\Sigma^{n+1}$ depending on an additional formal parameter~$u$ for $n\ge0$ by the following relations:
\begin{align}\label{eq:cT1x}
\cT_n(z,u;z_{\set n})&=\sum_{k=1}^\infty\frac1{k!}\prod_{i=1}^k
\Bigl(\restr{\tilde z_{i}}{z}u\hbar\cS(u\hbar\partial_{\tilde x_i})\tfrac{1}{d\tilde x_{i}}\Bigr)
\bigl(\omega_{k+n}(\tilde z_{\set{k}},z_{\set{n}})-
\delta_{n,0}\delta_{k,2}\tfrac{d\tilde x_{1}d\tilde x_{2}}{(\tilde x_{1}-\tilde x_{2})^2}\bigr)
\\\notag&\qquad\qquad+\delta_{n,0}u\bigl(\cS(u\hbar\partial_{x})-1\bigr)y,
\\\label{eq:cW1}
\cW_n(z,u;z_{\set n})&=\frac{dx}{u\hbar}e^{\cT_0(z,u)}
\sum_{\substack{\set{n}=\sqcup_{\alpha} J_\alpha
\\J_\alpha\ne\emptyset}}
\prod_{\alpha}\cT_{|J_\alpha|}(z,u;z_{J_\alpha}).
\end{align}
Here $\tilde x_i= x(\tilde z_i)$, $x=x(z)$, $y=y(z)$.
Note that in order to shorten the notation here and below we routinely omit $\hbar$ in the list of variables, though all newly introduced objects do depend on it.

These differentials possess the following properties implied directly by the definition:
\begin{itemize}
\item $\cW^{(g)}_n$ is a polynomial combination of the differentials $dx$, $dy$, $\omega^{(g')}_{n'}$, and their derivatives (in particular, the function~$y$ enters with derivatives of positive order only, and the explicit dependence of $\tilde x_1, \tilde x_2$ in the term with $\delta_{n,0}\delta_{k,2}$ disappears as well). As a consequence, it is a global meromorphic $(n+1)$-differential on~$\Sigma^{n+1}$.
\item We have $\cW^{(0)}_0=[\hbar^{-1}]\cW_0=\frac{dx}{u}$, and for $(g,n)\ne (0,0)$ the dependence of~$\cW^{(g)}_n$ on~$u$ is polynomial.
\item For each $(g,n)\ne(0,0)$ we have $[u^0]\cW^{(g)}_n(z,u;z_{\set n})=\omega^{(g)}_{n+1}(z,z_{\set n})$.
\item Let us represent $\cW^{(g)}_n$ for $(g,n)\ne(0,0)$ as
\begin{equation}
\cW^{(g)}_n(z,u;z_{\set n})=\omega^{(g)}_{n+1}(z,z_{\set n})+\overline\cW^{(g)}_n(z,u;z_{\set n}),
\end{equation}
where $\overline\cW^{(g)}_n$ is the contribution of the terms with positive exponents of~$u$. Then the expression for $\overline\cW^{(g)}_n$ involves the differentials $\omega^{(g')}_{n'}$ with $2g'-2+n'<2g-1+n$ only.
\end{itemize}

We refer to~\cite[Section 2.3.1]{BDKS2} and to Section~\ref{sec:identification} below for the examples.

\begin{definition}[Generalized topological recursion] \label{def:TRgeneral}
The \emph{differentials of generalized topological recursion} $\omega^{(g)}_n$, $g\geq 0$, $n\geq 1$ for the initial data $(\Sigma,dx,dy,B,\cP)$ are defined in the unstable cases by
\begin{equation}
\omega^{(0)}_1(z)=y(z)\;dx(z),\qquad \omega^{(0)}_2(z_1,z_2)=B(z_1,z_2),
\end{equation}
and for $2g-2+n>0$ they are given by
\begin{equation}
\label{eq:newproj}
\omega^{(g)}_{n}(z,z_{\set{n-1}})=\sum_{q\in\cP}\res\limits_{\tilde z=q}
\bar\omega^{(g)}_{n}(\tilde z,z_{\set{n-1}})\int\limits^{\tilde z}B(\cdot,z),
\end{equation}
where
\begin{equation}\label{eq:newTR}
\bar \omega^{(g)}_{n}(z,z_{\set{n-1}})=
-\sum_{r\ge1}\bigl(-d\tfrac{1}{dy}\bigr)^r[u^r]\overline\cW^{(g)}_{n-1}(z,u;z_{\set {n-1}}).
\end{equation}
\end{definition}

\begin{remark}
The differential $\bar \omega^{(g)}_{n}$, $g\geq 0$, $n\geq 1$, $2g-2+n\ge0$, defined by~\eqref{eq:newTR} is globally defined and meromorphic. Its possible poles with respect to the argument~$z$ are at \nonregular{} points, both \xvital{} and \yvital{},
as well as at the diagonals $z=z_i$. The differential~$\omega^{(g)}_{n}$ produced by the transformation~\eqref{eq:newproj} has the same principal parts of poles at the \xvital{}  \nonregular{} points, but is holomorphic at the other poles of $\bar \omega^{(g)}_{n}$.
\end{remark}

\begin{remark}
In the algebraic case, that is, if $\Sigma$ is compact and $B$ is the standard Bergman kernel normalized on $\mathfrak{A}$-cycles, the principal parts of the poles of $\omega^{(g)}_{n}$ together with the vanishing of ~$\mathfrak A$-periods determine these differentials uniquely. 
\end{remark}

\begin{remark}
Note that the relation defining the principal parts of the poles of~$\omega^{(g)}_{n}$ can be written in an equivalent form:
\begin{equation}\label{eq:mainTRrelation}
\sum_{r\ge0}\bigl(-d\tfrac{1}{dy}\bigr)^r[u^r]\cW^{(g)}_{n-1}(z,u;z_{\set {n-1}})\quad \text{is holomorphic at $z=q$ for $q\in\cP$,}
\end{equation}
for all $(g,n)$.
\end{remark}

\begin{remark}
If one multiply $dx$, or, equivalently $dy$, by a constant $\alpha$, then from the definition it immediately follows that the differentials \eqref{eq:newproj} are multiplied by $\alpha^{2-2g-n}$. This is a well-known property of the original topological recursion that is also valid for the generalized setup. 
\end{remark}

\begin{example}
	Below we provide several examples of differentials $\bar\omega_n^{(g)}$ for small values of $(g,n).$
	\begin{align}
		\bar\omega_3^{(0)}(z,z_{\set{2}}) &= d\,\frac{B(z,z_1)B(z,z_2)}{dx\, dy},\\
		\bar\omega_4^{(0)}(z,z_{\set{3}}) &= d\biggl(-\partial_y \frac{B(z,z_1)B(z,z_2)B(z,z_3)}{dx^2\, dy} + \frac{\omega_3^{(0)}(z,z_1,z_2)B(z,z_3)}{dx\,dy} 
		\\ \notag & \qquad  + \frac{\omega_3^{(0)}(z,z_1,z_3)B(z,z_2)}{dx\,dy} + \frac{ \omega_3^{(0)}(z,z_2,z_3)B(z,z_1)}{dx\,dy}\biggr),\\
		\bar\omega_1^{(1)}(z) &= d\bigg(\frac12\frac{\widetilde{B}(z,z)}{dx\,dy}-\frac1{24}\partial_y^2\partial_xy\bigg).
	\end{align}
	Here $\widetilde{B}(z_1,z_2) = B(z_1,z_2) - \frac{dx_1dx_2}{(x_1-x_2)^2}.$
\end{example}

\begin{proposition}\label{prop:symmetry}
The differentials of generalized topological recursion $\omega^{(g)}_{n}$ (in the sense of Definition~\ref{def:TRgeneral}, that is, defined by~\eqref{eq:newproj}--\eqref{eq:newTR}) are symmetric in all $n$ arguments. 
\end{proposition}

\begin{theorem}\label{th:CEO-compatibility} We have:
\begin{itemize}
	\item If $q\in\cP$ is a simple zero of~$dx$, and $dy$ is holomorphic and nonvanishing at~$q$ then the differentials given by~\eqref{eq:CEO} and \eqref{eq:newTR} differ by a summand holomorphic in $z$ at $q$. 
	\item 
	(An immediate corollary of the first statement.) If the initial spectral curve data satisfies the standard setup of the CEO topological recursion
	 %nondegeneracy condition required for CEO topological recursion 
	 and the set of \xvital{} points is chosen to be the set of zeroes of~$dx$, then the generalized topological recursion of Definition~\ref{def:TRgeneral} produces the same differentials as the original CEO topological recursion.
\end{itemize}
\end{theorem}

Proposition~\ref{prop:symmetry} and Theorem~\ref{th:CEO-compatibility} are proved in Sect.~\ref{sec:xyswap} and~\ref{sec:loop}.

\section{\texorpdfstring{$x-y$}{x-y} duality}\label{sec:xyswap}

\subsection{Recollection}
Given a system of symmetric multi-differentials $\{\omega^{(g)}_n\}_{(g,n)\ne(0,1)}$ on $\Sigma$, the $x-y$ duality transformation~\cite{borot2023functional,hock2022simple,ABDKS1,ABDKS-logTR-xy} produces a new one denoted by
$\{\omega^{\vee,(g)}_n\}_{(g,n)\ne(0,1)}$ on the same curve~$\Sigma$ and defined by a closed explicit expression involving two additional differentials~$dx$ and~$dy$. We review the definition of this transformation below. 

For the correctness of this definition we assume that the differentials~$dx$ and~$dy$ are meromorphic and also we require that $\omega^{(g)}_n$'s are globally defined, meromorphic, and have no poles on the diagonals for $(g,n)\ne(0,2)$. The differential $\omega^{(0)}_2$ is also meromorphic but it has a double order pole on the diagonal such that for any local coordinate $z$ the form $\omega^{(0)}_2(z_1,z_2)-\frac{dz_1dz_2}{(z_1-z_2)^2}$ is regular on the diagonal $z_1=z_2$.

Note that we do not assume any topological recursion for the differentials $\{\omega^{(g)}_n\}_{(g,n)\ne(0,1)}$, and \emph{a priori} these multi-differentials have no relation to the differentials $dx$ and $dy$; the latter ones are given as an extra data in the $x-y$ duality transformation. 

In certain cases, it might be convenient to set additionally $\omega^{(0)}_1=y\,dx$, $\omega^{\vee,(0)}_1=x\,dy$. However, we do not assume that $x$ and~$y$ are themselves meromorphic so that the meromorphy condition for $(g,n)=(0,1)$ might break.

\subsection{Explicit formulas}
Recall that 
\begin{equation}
\omega_n=\sum_{\substack{g\ge0,\\(g,n)\ne(0,1)}}\hbar^{2g-2+n}\omega^{(g)}_n,
\qquad
\omega^\vee_n=\sum_{\substack{g\ge0,\\(g,n)\ne(0,1)}}\hbar^{2g-2+n}\omega^{\vee,(g)}_n,
\end{equation}
and define the extended $n$-differentials $\bW_n=\sum_{g=0}^\infty\hbar^{2g-2+n}\bW^{(g)}_n$~\cite{ABDKS1}, where
\begin{multline} \label{eq:def-bW}
	\bW_n = \bW_n(z_{\llbracket n \rrbracket},u_{\llbracket n\rrbracket}) \coloneqq
	\prod_{i=1}^n \frac{dx_i}{u_i\hbar} e^{u_i(\cS(u_i\hbar\partial_{x_i})-1) y_i}
	\sum_{\Gamma} \frac{1}{|\mathrm{Aut}(\Gamma)|}
	\\
	\prod_{e\in E(\Gamma)}
	\prod_{j=1}^{|e|}\restr{(\tilde u_j, \tilde x_j)}{ (u_{e(j)},x_{e(j)})} \tilde u_j\hbar  \cS(\tilde u_j \hbar  \partial_{\tilde x_j})  \frac{\tilde \omega_{|e|}(\tilde z_{\llbracket |e|\rrbracket})}  {\prod_{j=1}^{|e|} {d\tilde x_j}}.
\end{multline}
Here $x_i = x(z_i)$, $\tilde x_i = x(\tilde z_i)$, $y_i=y(z_i)$, and $\tilde \omega_n = \omega_n$ unless $n=2$ and $e(1)=e(2)$. In the latter case $\tilde \omega_2 = \omega_2 - dx_1dx_2/(x_1-x_2)^2$. The sum is taken over all connected graphs $\Gamma$ with $n$ labeled vertices and multiedges of index $\geq 1$. 
The index of a multiedge $e$ is the number of its ``legs'' and we denote it by $|e|$. For a multiedge $e$ with index $|e|$  we control its attachment to the vertices by the associated map $e\colon \llbracket |e| \rrbracket \to
\llbracket n \rrbracket
$ that we denote also by $e$, abusing notation (so $e(j)$ is the label of the vertex to which the $j$-th ``leg'' of the multiedge $e$ is attached).

Analogously, define $\bW^\vee_n=\sum_{g=0}^\infty\hbar^{2g-2+n}\bW^{\vee,(g)}_n$
as 
\begin{multline} \label{eq:def-bW-vee}
	\bW^{\vee}_n = \bW^{\vee}_n(z_{\llbracket n \rrbracket},v_{\llbracket n\rrbracket}) \coloneqq
	\prod_{i=1}^n \frac{dy_i}{v_i\hbar} e^{v_i (\cS(v_i\hbar \partial_{y_i})-1)x_i}
	\sum_{\Gamma} \frac{1}{|\mathrm{Aut}(\Gamma)|}
	\\
	\prod_{e\in E(\Gamma)}
	\prod_{j=1}^{|e|}\restr{(\tilde v_j, \tilde y_j)}{ (v_{e(j)},y_{e(j)})} \tilde v_j \hbar \cS(\tilde v_j \hbar  \partial_{\tilde y_j}) \frac{\tilde \omega^{\vee}_{|e|}(\tilde z_{\llbracket |e|\rrbracket})}  {\prod_{j=1}^{|e|} {d\tilde y_j}}.
\end{multline}
Here we use the same conventions and  $\tilde \omega^\vee_n = \omega^\vee_n$ unless $n=2$ and $e(1)=e(2)$. In the latter case $\tilde \omega^\vee_2 = \omega^\vee_2 - dy_1dy_2/(y_1-y_2)^2$.

The differentials $\bW_n$, resp.~$\bW^\vee_n$, are symmetric in $n$ pairs of variables $(z_i,u_i)$, resp.~$(z_i,v_i)$, where $z_i$ is a point of the $i$-th factor in~$\Sigma^n$ and $u_i$, resp.~$v_i$, is a formal parameter.
Note also that we have the following connection between $\bW_n$ and $\cW_n$ defined in~\eqref{eq:cW1}:
\begin{equation}
	\cW_{n}(z,u,z_{\set{n}}) = \restr{u_{\set{n}}}{0} \restr{z_{n+1}}{z} \restr{u_{n+1}}{u} \bW_{n+1}(z_{\set{n+1}},u_{\set{n+1}}).
\end{equation}

% With the exception of $\bW^{(0)}_1=\frac{1}{u_1}$  corresponding to leading term of the trivial graph, $\bW^{(g)}_n$ is a polynomial in $u_1,\dots,u_n$ whose coefficients are meromorphic functions on $\Sigma^n$ represented as finite symbolic expressions in $\omega^{(g)}_n$'s, $x'$ and $y'$.

\begin{definition} For a given system of multi-differentials $\{\omega^{(g)}_n\}_{(g,n)\ne(0,1)}$ and two differentials $dx$ and $dy$, the $x-y$ dual system of multi-differentials $\{\omega^{\vee,(g)}_n\}_{(g,n)\ne(0,1)}$ is defined as 
\begin{equation}\label{eq:xyrel}
\omega^\vee_n(z_{\set n})=(-1)^n
\left(\prod_{i=1}^n\sum_{r=0}^\infty \bigl(-d_i\tfrac{1}{dy_i}\bigr)^{r}[u_i^r]\right)
\bW_n(z_{\set n},u_{\set n}).
\end{equation}
\end{definition}

\begin{remark}[\cite{ABDKS3}]
The inverse transformation is given by similar expression with the roles of $x$ and $y$ swapped
\begin{equation}
\omega_n(z_{\set n})=(-1)^n
\left(\prod_{i=1}^n\sum_{r=0}^\infty \bigl(-d_i\tfrac{1}{dx_i}\bigr)^{r}[v_i^r]\right)
\bW^\vee_n(z_{\set n},v_{\set n}).
\end{equation}	
\end{remark}

\begin{remark}[\cite{ABDKS1}] \label{rem:Inductive-x-y-swap}
The summation over graphs can be replaced by a two-index induction which will be very efficient in the proof of loop equation (see Theorem \ref{Thm:loopeq}). Set $\omega_{m,0}=\omega_m$, $m\geq 1$, and define inductively

\begin{align}\label{eq:cTmn}
	\cT_{m,n}(z,u;z_{\set{m}};z_{\set{m+n}\setminus \set{m}}) & \coloneqq
	\sum_{k=1}^\infty\frac1{k!}\prod_{i=1}^k
	\Bigl(\restr{\tilde z_{i}}{z}u\hbar\cS(u\hbar\partial_{\tilde x_i})\tfrac{1}{d\tilde x_{i}}\Bigr)
	\\ \notag & \qquad 
		\bigl(\omega_{k+m,n}(\tilde z_{\set{k}},z_{\set{m}};z_{\set{n+m}\setminus \set{m}})-
	\delta_{m,0}\delta_{n,0}\delta_{k,2}\tfrac{d\tilde x_{1}d\tilde x_{2}}{(\tilde x_{1}-\tilde x_{2})^2}\bigr)
	\\\notag&\quad
		+\delta_{m,0}\delta_{n,0}\,u\bigl(\cS(u\hbar\partial_x)-1\bigr)y,
	\\ 	\label{eq:cWmn}
	\cW_{m,n}(z,u;z_{\set{m}};z_{\set{m+n}\setminus \set{m}}) & \coloneqq \frac{dx}{u\hbar}e^{\cT_0(z,u)} \sum_{\substack{\set{m+n}=\sqcup_\alpha K_\alpha \\
			K_\alpha\ne\varnothing\\ I_\alpha=K_\alpha\cap \set{m} \\ J_\alpha=K_\alpha\setminus I_\alpha}} 
	\prod_{\alpha} \cT_{|I_\alpha|,|J_\alpha|}(z,u; z_{I_\alpha};z_{J_\alpha}),
\\\label{eq:cWmn-to-omega}
\omega_{m,n+1}(z_{\set m};z,z_{\set {m+n}\setminus \set{m}})&\coloneqq -\sum_{r=0}^\infty \bigl(-d\tfrac{1}{dy}\bigr)^{r}[u^r]\cW_{m,n}(z,u;z_{\set{m}};z_{\set{m+n}\setminus \set{m}}),
\end{align}
so that $\cT_{m,0}=\cT_m$ given in~\eqref{eq:cT1x} and $\cW_{m,0}=\cW_m$ given in~\eqref{eq:cW1}, $m\geq 1$. Then, $\omega^\vee_n=\omega_{0,n}$, $n\geq 1$.

The inverse transformation can be represented by a similar induction: 
\begin{align}\label{eq:cTmn-vee}
	\cT^\vee_{m,n}(z,v;z_{\set{m}};z_{\set{m+n}\setminus \set{m}}) & \coloneqq
	\sum_{k=1}^\infty\frac1{k!}\prod_{i=1}^k
	\Bigl(\restr{\tilde z_{i}}{z}v\hbar\cS(v\hbar\partial_{\tilde y_i})\tfrac{1}{d\tilde y_{i}}\Bigr)
	\\ \notag & \qquad 
	\bigl(\omega_{m,k+n}(z_{\set{m}};\tilde z_{\set{k}},z_{\set{n+m}\setminus \set{m}})-
	\delta_{m,0}\delta_{n,0}\delta_{k,2}\tfrac{d\tilde y_{1}d\tilde y_{2}}{(\tilde y_{1}-\tilde y_{2})^2}\bigr)
	\\\notag&\quad
	+\delta_{m,0}\delta_{n,0}\,v\bigl(\cS(v\hbar\partial_y)-1\bigr)x,
	\\ 	\label{eq:cWmn-vee}
	\cW^\vee_{m,n}(z,u;z_{\set{m}};z_{\set{m+n}\setminus \set{m}}) & \coloneqq \frac{dy}{v\hbar}e^{\cT_0(z,v)} \sum_{\substack{\set{m+n}=\sqcup_\alpha K_\alpha \\
			K_\alpha\ne\varnothing\\ I_\alpha=K_\alpha\cap \set{m} \\ J_\alpha=K_\alpha\setminus I_\alpha}} 
	\prod_{\alpha} \cT_{|I_\alpha|,|J_\alpha|}(z,v; z_{I_\alpha};z_{J_\alpha}),
	\\\label{eq:cWmn-vee-to-omega}
	\omega_{m+1,n}(z,z_{\set m};z_{\set {m+n}\setminus \set{m}})&\coloneqq -\sum_{r=0}^\infty \bigl(-d\tfrac{1}{dx}\bigr)^{r}[v^r]\cW^\vee_{m,n}(z,u;z_{\set{m}};z_{\set{m+n}\setminus \set{m}}).
\end{align}
These formulas allow to inductively obtain $\{\omega_m=\omega_{m,0}\}_{m\geq 1}$ starting from $\{\omega_n^\vee = \omega_{0,n}\}_{n\geq 1}$. 
\end{remark}

\subsection{Formal properties of the \texorpdfstring{$x-y$}{x-y} duality relation} 

\begin{lemma}\label{lem:xyregular}
Let $q\in\Sigma$ be a \regular{} point in a sense of Definition~\ref{def:singpoint}. Then it is regular for all differentials~$\omega^{\vee,(g)}_n$, $(g,n)\ne(0,1)$, if and only if it is regular for all differentials~$\omega^{(g)}_n$, $(g,n)\ne(0,1)$.
\end{lemma}

\begin{proof}
Let $dx\sim z^{r-1}dz$, $dy\sim z^{s-1}dz$ at the considered point. In the case $r=s=1$ all derivatives entering \eqref{eq:def-bW} and \eqref{eq:xyrel} (or \eqref{eq:cTmn}--\eqref{eq:cWmn-to-omega})
obviously preserve the space of holomorphic functions, so we only need to consider the case $r+s\le0$. We introduce filtration in the space of Laurent expansions in~$z_i,u_i$ by setting $\deg(z_i)=\deg(dz_i)=1$, $\deg(u_i)=-s$. Then $\deg x_i=r$ and the very form of the expression for $\bW^{(g)}_n$, $(g,n)\not=(0,1)$, implies that all the terms of its Laurent expansion in~$z_i,u_i$ are of non-negative degree with respect to the filtration (though it might contain monomials with negative exponents of $z_i$). The transformation~\eqref{eq:xyrel} expressing $\omega^{\vee,(g)}_n$ in terms of $\bW^{(g)}_n$ also preserves this filtration, due to the condition $r+s\leq 0$. Thus, the degree of $\omega^{\vee,(g)}_n$ with respect to the filtration is non-negative. Since it does not involve dependence on~$u_i$, we conclude that it is holomorphic.

The ``only if'' part of Lemma follows from similar arguments applied to the inverse transformation of $x-y$ duality.
\end{proof}

Let us stress that up to this point in our discussion of the $x-y$ duality we assumed no topological recursion for the differentials $\omega^{(g)}_n$. Recall, however, Def.~\ref{def:TRgeneral}. The following lemma serves as a main motivation for this definition.

\begin{lemma}\label{lem:xyrecursion}
Requirement~\eqref{eq:mainTRrelation} defining the principal part of the poles of $\omega^{(g)}_n$ at $q\in\cP$ is equivalent to the requirement that the $x-y$ dual differentials $\omega^{\vee}_n$ are holomorphic at~$q$ (with respect to any of their arguments).
\end{lemma}

\begin{proof}
As we mentioned in Rem.~\ref{rem:Inductive-x-y-swap}, $\cW_n=\cW_{n,0}$. Thus, from \eqref{eq:cWmn-to-omega} we see that the differential given by~\eqref{eq:newTR} coincides with 
\begin{equation}
\bar{\omega}^{(g)}_{n+1}(z,z_{\set{n}})=\omega^{(g)}_{n+1}(z,z_{\set{n}})+\omega^{(g)}_{n,1}(z_{\set{n}};z).
\end{equation}
Therefore, Eq.~\eqref{eq:mainTRrelation} is equivalent to the requirement that $\omega^{(g)}_{n,1}(z;z_{\set{n}})$ is holomorphic at $z=q$ for $q\in\cP$ for all $(g,n)$. By the combinatorics of Equations~\eqref{eq:cTmn}--\eqref{eq:cWmn-to-omega}, we see that the latter requirement is equivalent to the one that  $\omega^{(g)}_{0,n}=\omega^{\vee,(g)}_n$ are holomorphic as well at $q\in\cP$ in all its arguments for all $(g,n)\neq (0,1)$.
\end{proof}

\subsection{\texorpdfstring{$x-y$}{x-y} duality for the generalized topological recursion} 

Now we do assume that the differentials $\omega^{(g)}_n$ are obtained by generalized topological recursion and proceed to the proof of Proposition~\ref{prop:symmetry}.

\begin{proof}[Proof of Proposition~\ref{prop:symmetry}]
%Extending the arguments in the above proof, 
We proceed by induction on the negative Euler characteristic $2g-2+n$. 
Let us represent the extended differential $\bW^{(g)}_n$ in the form
\begin{equation}
\bW^{(g)}_n(z_{\set n},u_{\set n})=\omega^{(g)}_n(z_{\set n})+\overline\bW^{(g)}_n(z_{\set n},u_{\set n}),
\end{equation}
where $\overline\bW^{(g)}_n(z_{\set n},u_{\set n})$ is an expression that involves all $u$-variables with the strictly positive sum of their exponents. All differentials $\omega^{(g')}_{n'}$ entering
the expression for $\overline\bW^{(g)}_n(z_{\set n},u_{\set n})$ are assumed being computed in the previous steps of recursion and known to be symmetric in all arguments. Denote
\begin{equation}\label{eq:ttomega}
\tilde{\tilde \omega}^{(g)}_n=-
[\hbar^{2g-2+n}]\left(\prod_{i=1}^n\sum_{r=0}^\infty \bigl(-d_i\tfrac{1}{dy_i}\bigr)^{r}[u_i^r]\right)
\overline\bW_n(z_{\set n},u_{\set n}).
\end{equation}
On one hand, by induction hypothesis, this form is symmetric in all its arguments. On the other hand, we have, by~\eqref{eq:xyrel},
\begin{equation}
\tilde{\tilde \omega}^{(g)}_n=\omega^{(g)}_n-(-1)^n\omega^{\vee,(g)}_n,
\end{equation}
and the form $\omega^{(g)}_n$ can be recovered from $\tilde{\tilde \omega}^{(g)}_n$ as
\begin{equation}
\omega^{(g)}_n(z_{\set n})=\left(\prod_{i=1}^n\sum_{q\in\cP}\res\limits_{\tilde z_i=q}\right)
\tilde{\tilde \omega}^{(g)}_n(\tilde z_{\set n})
\prod_{i=1}^n\int\limits^{\tilde z_i}B(\cdot,z_i).
\end{equation}
This shows that $\omega^{(g)}_n$ is symmetric as well.
\end{proof}

Globalizing the local considerations above, we obtain the main result of this section.

\begin{theorem}\label{th:xyswap} %In the algebraic case, the $x-y$ swap transformation preserves topological recursion. 
	
	%Namely, 
	Assume that $\Sigma$ is algebraic (that is, compact) and $B$ is the canonical Bergman kernel fixed by the condition of vanishing of $\mathfrak A$-periods. Then the differentials that are $x-y$ dual to the generalized topological recursion differentials with the initial data $(\Sigma,dx,dy,B,\cP)$, solve the generalized topological recursion for the dual initial data $(\Sigma,dy,dx,B,\cP^\vee)$.
\end{theorem}

\begin{proof}
By Lemma~\ref{lem:xyregular}, the only poles that $\omega^{\vee,(g)}_n$ can posses are \nonregular{} points which are not \xvital, that is, the \yvital{} \nonregular{} points. By the relations of the inverse duality transformation, the principal parts of the poles are defined by the same relations as for the generalized topological recursion with the $x-y$-dual spectral curve data $(\Sigma,dy,dx,B,\cP^\vee)$, since the differentials $\omega^{(g)}_n=\omega^{\vee\vee,(g)}_n$ are holomorphic at $q\in\cP^\vee$. Finally, since $\tilde{\tilde \omega}^{(g)}_n$ defined by~\eqref{eq:ttomega} are total differentials, vanishing of $\mathfrak A$-periods of $\omega^{\vee,(g)}_n$ follows from the same property of $\omega^{(g)}_n$. Thus, the differentials $\omega^{\vee,(g)}_n$ satisfy all relations of generalized topological recursion.
\end{proof}

\begin{remark} In the setup of the original CEO topological recursion an analog of this theorem also holds; it was conjectured in~\cite[Conjecture 3.14]{borot2023functional} and proved in~\cite[Theorem 1.14]{ABDKS1}. Note that by Theorem~\ref{th:CEO-compatibility}, the latter result is a special case of Theorem~\ref{th:xyswap}. 
\end{remark}

\section{Loop equations and compatibility with original topological recursion}\label{sec:loop}

\subsection{Towards the comparison of topological recursions}
The computation of the differentials of the generalized topological recursion is realized in two steps: 
\begin{itemize}
	\item At the first step we determine the principal parts of the poles of $\omega^{(g)}_n$ at the (\xvital) \nonregular{} points.
	\item At the second step we recover the differential $\omega^{(g)}_n$ from the principal parts of its poles.
\end{itemize}
 The loop equations introduced below provide yet another equivalent way to realize the first step, that is, to identify the principal part of the pole of $\omega^{(g)}_n$ at a \nonregular{} point. Loop equations considered in this section are applicable if $dy$ is holomorphic and non-vanishing and $dx$ has a (possibly multiple) zero at the corresponding \nonregular{} point. 
 
 In the case when the corresponding zero of $dx$ is simple, the loop equations are proved also to be equivalent to the original CEO topological recursion relations~\eqref{eq:CEO} (see also~\cite{BEO-loop,BS-blobbed}). Therefore, as a byproduct, we obtain that in the case of standard nondegenerate initial data as in Sect.~\ref{sec:nondegenerateTR}, the generalized topological recursion given by Def.~\ref{def:TRgeneral} coincides with the original CEO topological recursion reviewed in Sect.~\ref{sec:nondegenerateTR}. 

\subsection{Loop equations}
Let $dx$ be a meromorphic differential on~$\Sigma$ and $q\in\Sigma$ be a point such that $dx$ is holomorphic and has a zero of order $r-1$ at $q$ for some $r\ge2$. Denote by $x$ the local primitive of $dx$ with the integration constant such that $x$ vanishes at~$q$ (it is defined in some neighborhood of~$q$). Then we can choose a local coordinate~$\zeta$ on $\Sigma$ centered at $q$ such that $x(\zeta)=\zeta^r$.

\begin{definition}
We denote by $\Xi_q$ the space of germs of meromorphic $1$-forms on $\Sigma$ at the point $q$ whose Laurent expansions contain no monomials of the form $\frac{d\zeta}{\zeta^{kr+1}}$, $k\ge0$.
\end{definition}

The following lemma is an easy exercise.
\begin{lemma}\label{lem:Xi}
The space $\Xi_q$ admits the following three equivalent invariant descriptions:
\begin{itemize}
\item $\Xi_q$ is spanned by meromorphic differentials of the kind $\bigl(d\frac{1}{dx}\bigr)^k\alpha$ where $k\ge 0$ and $\alpha$ is holomorphic;
\item $\eta\in\Xi_q$ iff for any holomorphic function (a regular power series) $f$ we have $\res_{z=q}f(x(z))\,\eta(z)=0$;
\item $\eta\in\Xi_q$ iff $x_*\eta$ is holomorphic.
\end{itemize}
\end{lemma}
Recall that the push-forward or trace homomorphism $x_*$ takes a given $1$-form~$\eta$ on $\Sigma$ to the $1$-form on the target of~$x$ regarded as a degree $r$ ramified covering obtained by summation of values of $\eta$ at all $r$ preimages of a given point.

\begin{theorem}%[Loop equations]
	\label{Thm:loopeq}
Let $\{\omega^{(g)}_n\}$ be the differentials of the generalized topological recursion (in the sense of Def.~\ref{def:TRgeneral}), and $q\in\cP$ be a \xvital{} \nonregular{} point such that both $dx$ and $dy$ are holomorphic at this point and $dy\ne0$. Then for any $k\ge0$ the following relation holds true:
\begin{equation} \label{eq:(k-1)-loop}
[u^k]e^{u y}\cW^{(g)}_n(z,u)\in\Xi_q,
\end{equation}
where the extended differentials $\cW^{(g)}_n(z,u)=\cW^{(g)}_n(z,u;z_{\set{n}})$ are defined by \eqref{eq:cT1x}--\eqref{eq:cW1} and $y$ is any (local) primitive of $dy$. This relation holds identically in $z_1,\dots,z_n$.%. and is called the \emph{loop equations} (cf. Def.~).
\end{theorem}

\begin{proof}
The extended differentials satisfy the following identity proved in~\cite{ABDKS1} in the context of general $x-y$ duality:
\begin{equation}\label{eq:looprel}
e^{u y}\cW_n(z,u)=e^{u y}\cW_{n,0}(z,u)=-\sum_{k=0}^\infty\bigl(-d\tfrac{1}{dx}\bigr)^ke^{u y}[v^k]\cW^\vee_{n,0}(z,v)+
\delta_{n,0}\frac{e^{u y}dx}{\hbar u},
\end{equation}
where the dual mixed extended differentials $\cW^\vee_{m,n}$ are defined by~\eqref{eq:cWmn-vee}. By construction, it involves dual differentials~$\omega^\vee_{n}$,~$dx$, and the derivatives of those with respect to~$y$. Therefore, by Lemma \ref{lem:xyrecursion}, it is holomorphic at $z=q$. We see that the loop equations are implied by the very form of~\eqref{eq:looprel}, by Lemma~\ref{lem:Xi}.
\end{proof}

\begin{definition} \label{eq:def-loop}
	Relation~\eqref{eq:(k-1)-loop} is called the \emph{$(k+1)$-loop equation}. In particular, for $k=0$, resp. $k=1$, it is also called the linear, resp. quadratic, loop equation.
\end{definition}

\begin{remark} We stress the difference between the loop equations as defined and discussed in this section and the so-called \emph{abstract loop equations} as in~\cite[Definition 2.14]{limits}. In general these two systems of equations are non-equivalent. 
\end{remark}

\subsection{Identification of the CEO and the generalized topological recursions}

\label{sec:identification}

\begin{proof}[Proof of Theorem~\ref{th:CEO-compatibility}]
Let us write expressions for $[u^k]e^{u y}\cW^{(g)}_n$ for $k=0$ and~$1$ corresponding to the linear and quadratic loop equations explicitly:
\begin{align}
[u^0]e^{u y}\cW^{(g)}_n&=\omega^{(g)}_{n+1}(z,z_{\set{n}}),
\\ [u^1]e^{u y}\cW^{(g)}_n&=
\frac{1}{2dx(z)}\Bigl(\omega^{(g-1)}_{n+2}(z,z,z_{\set{n}})+
\sum_{\substack{g_1+g_2=g\\J_1\sqcup J_2=\set{n}}}
\omega_{|J_1|+1}^{(g_1)}(z,z_{J_1})\omega_{|J_2|+1}^{(g_2)}(z,z_{J_2})\Bigr)
\\\notag &=y(z)\,\omega^{(g)}_{n+1}(z,z_{\set{n}})+\frac{1}{2dx(z)}\Bigl(\omega^{(g-1)}_{n+2}(z,z,z_{\set{n}})+
\\\notag &\qquad\qquad
\sum_{\substack{g_1+g_2=g,~J_1\sqcup J_2=\set{n}\\(g_i,|J_i|)\ne(0,0)}}
\omega_{|J_1|+1}^{(g_1)}(z,z_{J_1})\omega_{|J_2|+1}^{(g_2)}(z,z_{J_2})\Bigr).
\end{align}

If the point $q\in \Sigma$ is a simple zero of $dx$, that is, $r=2$, then the linear and quadratic loop equations determine the principal part of the pole of $\omega^{(g)}_{n+1}$ uniquely, and it is given explicitly by~\eqref{eq:CEO}, see e.g.~\cite{ABDKS1}. We see that in the setting of Sect.~\ref{sec:nondegenerateTR}, both versions of topological recursion produce differentials with equal principal parts of the poles at \nonregular{} points. Note that the differentials are uniquely reconstructed from their principal parts at these points via exactly the same formulas (cf. Equations~\eqref{eq:proj} and~\eqref{eq:newproj}). Therefore, the differentials themselves are equal.
\end{proof}

If $k>1$, then $[u^k]e^{u y}\cW^{(g)}_n$ has more complicated expression. In particular, it involves not only the differentials $\omega^{(g')}_{n'}$ themselves but also their derivatives with respect to~$x$. However, one can notice that it has a general form
\begin{equation}
[u^k]e^{u y}\cW^{(g)}_n=\frac{(y(z))^k}{k!}\omega^{(g)}_{n+1}(z,z_{\set{n}})+\dots,
\end{equation}
where the dots denote the terms involving the differentials $\omega^{(g')}_{n'}$ with $2g'-2+n'<2g-1+n$, that is, the terms that can be treated as already computed in the previous steps of recursion. This implies, for example, that for any $r\ge2$ the first $r$ loop equations (for $k=0,1,\dots,r-1$) determine uniquely the principal part of the pole of $\omega^{(g)}_{n+1}$, and thus are sufficient for the recursion, see details in~\cite{ABDKS1}. The higher loop equations (for $k\ge r$) are also satisfied but they are formal corollaries of the first~$r$ loop equations.

%\sasha{Loop equations for $s\neq 1$}

\section{Compatibility with other versions of topological recursion}

\label{sec:compatibility}

In addition to the original CEO topological recursion,  in the introduction we mentioned a variety of its extensions, where the setup is more general than the standard one and the formulas are adjusted to capture higher order zeros of $dx$ and singularities of $dy$. The goal of this section is to give a comparison of the generalized topological recursion and these extensions of the original CEO. 

\subsection{Compatibility with LogTR} \label{sec:LogTR}

Assume that given meromorphic differentials $dx$ and $dy$ on $\Sigma$ have no common zeroes. Then \nonregular{} points in the sense of Definition~\ref{def:singpoint} include not only zeroes of~$dx$ and~$dy$, but also their simple poles. This is the situation that was treated in~\cite{ABDKS-logTR-xy,ABDKS-log-sympl} via an extension of the original CEO topological recursion called the log topological recursion. Since the latter one was essentially derived from the idea of compatibility with the $x-y$ duality, we have the following statement:

\begin{theorem}
Assume that all zeroes of~$dx$ are simple and $dy$ is holomorphic and non-vanishing at the zeroes of~$dx$. Let the set $\cP$ of \xvital{} \nonregular{} points consist of the zeroes of $dx$ and all poles of $dy$ (they are simple and distinct from the poles of~$dx$). Then the generalized topological recursion differentials with this spectral curve data coincide with the differentials of the log topological recursion of \cite{ABDKS-logTR-xy}.
\end{theorem}

\begin{proof}
This theorem is essentially a reformulation of the main result in \cite{ABDKS-logTR-xy}. Its technical formulation is contained in the following computation:

\begin{lemma}[{\cite{ABDKS-logTR-xy}, an $x-y$ dual of Lemma 5.5}] 
Let $q\in\Sigma$ be a simple pole of $dy$ with the residue $\alpha^{-1}$ and $dx$ be holomorphic and non-vanishing at~$q$. Let $\{\omega^{(g)}_n\}$ be a system of differentials such that their $x-y$ dual differentials $\omega^{\vee,(g)}_n$ are all holomorphic at~$q$ (with respect to each argument). Then the original differentials have the following expansion at the point~$z=q$:
\begin{equation}\label{eq:log-contribution}
\omega^{(g)}_{n+1}(z,z_{\set n})=
\begin{cases}
(\text{holomorphic in }z),&n>0,\\
[\hbar^{2g}]\left(\dfrac{1}{\cS(\alpha\hbar\partial_x)}y(z)\right)dx(z)+(\text{holomorphic in $z$}),&n=0.
\end{cases}
\end{equation}
\end{lemma} 

By Lemma~\ref{lem:xyrecursion}, the contribution to the pole of the generalized topological recursion differentials (in the sense of Def.~\ref{def:TRgeneral}) at a point $q\in \Sigma$ as in Lemma is given by the right hand side of~\eqref{eq:log-contribution}. On the other hand, the same contribution of this pole is taken for LogTR differentials, by definition. The contributions of the poles at zeroes of~$dx$ for both recursions also agree by Theorem~\ref{th:CEO-compatibility}. Thereby, the differentials of both recursions have equal principal parts of the poles. Therefore, the differentials themselves are equal as well.
\end{proof}

\subsection{Generalized topological recursion in families} In order to proceed with the comparison with more general versions of topological recursion, we need an extra tool. Assume that the components of the initial data $(\Sigma,dx,dy,B,\cP)$ of the generalized topological recursion depend on an additional parameter. In this section, we discuss conditions ensuring that the corresponding family of differentials analytically depend on this additional parameter. It might happen that some number of \nonregular{} points collapse together for the limiting value of parameter; in this case the differentials of the generalized topological recursion have a limit and this limit coincides with the differentials of the limiting initial data. Roughly speaking, the condition for the compatibility with the limit is that locally the collapsing \nonregular{} points should be either all treated as \xvital{} or all of them are \yvital. The limiting point should be either \regular{} in the sense of Definition~\ref{def:singpoint} or it should be of the same kind as the colliding \nonregular{} points, \xvital{} or \yvital, respectively.

For example, assume that both $dx$ and $dy$ have a pole at the same point for the limiting parameter value. Such a point is regarded as \regular{} and does not contribute to the topological recursion. These poles can split under variation of the parameter, and the deformed differentials $dx$, $dy$ may have both poles of smaller order and zeroes nearby. For the compatibility with the limits we should treat all ``newborn'' \nonregular{} points either all as \xvital, or all as \yvital, even though they might include both simple zeroes of~$dx$ and simple zeroes of~$dy$.

In order to formulate the above considerations in a more formal way, we consider the following setting. Let $U\subset\Sigma$ be an open domain with a smooth boundary such that the given differentials $dx$, $dy$ are holomorphic and non-vanishing along $\partial U$. If $\Sigma$ is not compact, we assume that the closure of~$U$ is compact. For example, we can take for $U$ the union of small discs centered at those \nonregular{} points which are chosen as \xvital{} for the pair of differentials $(dx,dy)$.

\begin{theorem} \label{th:analytic}
	Assume that the given spectral curve data varies in an analytic family such that $dx$ and $dy$ are regular and nonvanishing along $\partial U$ for all parameter values. Let us take for the \xvital{} and \yvital{} points those \nonregular{} points that belong to~$U$ and to its complement, respectively, for all parameter values. In this setting, the family of TR differentials depend on parameters analytically.
\end{theorem}

\begin{proof}
	In the setting of the theorem, the residues entering recursion relation~\eqref{eq:newproj} can be replaced by contour integrals, and we have
	\begin{equation}
		\begin{aligned}
			\omega^{(g)}_{n+1}(z,z_{\set n})&=\frac{1}{2\pi \sqrt{-1}}\int\limits_{\tilde z\in \partial U}\Bigl(\tilde\omega^{(g)}_{n+1}(\tilde z,z_{\set n})
			\int^{\tilde z}\!\!B(\cdot,z)\Bigr)
			\\&=
			-\frac{1}{2\pi \sqrt{-1}}\int\limits_{\tilde z\in \partial U}\Bigl(
			B(\tilde z,z)
			\int^{\tilde z}\!\!\tilde\omega^{(g)}_{n+1}(\cdot,z_{\set n})\Bigr)
		\end{aligned}
	\end{equation}
	for $z,z_1,\dots,z_n\in\Sigma\setminus \bar U$, where $\tilde \omega^{(g)}_{n+1}$ is given by~\eqref{eq:newTR} and $\bar U=U\cup \partial U$. An expression on the second line is valid for any $U$ since~\eqref{eq:newTR} implies that $\tilde \omega^{(g)}_{n+1}$ is a total differential: its primitive can be expressed as
	\begin{equation}
		\int^{z}\!\!\tilde\omega^{(g)}_{n+1}(\cdot,z_{\set n})=
		\sum_{r\ge1}\bigl(-\partial_y\bigr)^{r-1}[u^r]\frac{\overline\cW^{(g)}_n(z,u;z_{\set n})}{dy}.
	\end{equation}
	These explicit integral formulas imply that the generalized topological recursion differentials $\omega^{(g)}_n$ are holomorphic on $(\Sigma\setminus \bar U)^n$ and depend analytically on parameters, since all differentials involved are regular on the integration contour. We conclude that the limit of the generalized topological recursion differentials coincides with the generalized topological recursion differentials of the limiting spectral curve data on an open domain $(\Sigma\setminus \bar U)^n\subset \Sigma^n$, and hence, this is true for the whole $\Sigma^n$.
\end{proof}

\begin{example}\label{ex:rsfirst}
	Consider the rational spectral curve with
	\begin{equation}
		x=z^r+\tau\,z,\qquad y=z^s,\qquad r\ge2,\quad -r<s<0\text{ or }s=1.
	\end{equation}
	For $\tau\ne 0$ the CEO topological recursion can be applied. Its differentials have poles at the critical points of the function~$x$. They are defined by the equation $r z^{r-1}+\tau =0$. These differentials have a limit as $\tau\to 0$ and the limiting differentials are the differentials of the generalized topological recursion with the limiting spectral curve $(x,y)=(z^r,z^s)$ with the unique \xvital{} point $z=0$.
	
	The set of \yvital{} \nonregular{} points is empty for all~$\tau$ and the dual differentials are trivial $\omega^{\vee,(g)}_n=0$ for $2g-2+n>0$ identically in~$\tau$. The generalized topological recursion differentials corresponding to the spectral curve $(z^r,z^s)$ are discussed in more details in Sect.~\ref{sec:rs-examples}.
\end{example}

\subsection{Compatibility with the Chekhov--Norbury setup}

Chekhov and Norbury applied the original CEO topological recursion in a bit more general setup~\cite{CN}. Namely, they still allowed $x$ to have only non-degenerate critical points, but at each critical point of $x$ they allowed that either $y$ has a simple pole or $y$ is regular and $dy$ is non-vanishing. Note that this means that $ydx$ is still regular at each critical point of $x$. The formulas of the recursion are then exactly the same as the ones for the original CEO topological recursion~\eqref{eq:CEO}. However, the argument as in the proof of Theorem~\ref{th:CEO-compatibility} is not directly applicable, since in this case the \emph{abstract loop equations}  in the sense of~\cite[Definition 2.14]{limits} are no longer compatible with the loop equation in Section~\ref{sec:loop}. 
However, we still have 

\begin{theorem} \label{thm:CN}
	In the Chekhov--Norbury setup the differentials of the original topological recursion coincide with the differentials of the generalized topological recursion, where the set $\cP$ is chosen to be the set of the critical points of $x$. 
\end{theorem}

\begin{proof} This theorem follows from Theorem~\ref{th:analytic} above and~\cite[Theorem 5.8]{limits}. 
	
We can choose a $U\subset \Sigma$ to be a union of open disks around $\cP$ that contains no other \nonregular{} points, and deform the pair $(x_0,y_0)\coloneqq (x,y)$ to $(x_\epsilon,y_\epsilon)$, $\epsilon$ is a parameter of deformation, such that 
\begin{itemize}
	\item $x_\epsilon = x$;
	\item  in each connected component of $U$ either $y_\epsilon$ or $y_\epsilon^{-1}$ remains to be a local coordinate;
	\item the poles of $y_\epsilon$ are disjoint from $\cP$ for $\epsilon \not=0$ 
\end{itemize}	
(it is possible to consider such a deformation locally on a small neighborhood of $U\cup \partial U$). 

In this setup, the generalized topological recursion differentials coincide with the original CEO topological recursion differentials computed for $\epsilon\not=0$ by Theorem~\ref{th:CEO-compatibility}. On the other hand, both systems of differentials (identified for $\epsilon\not=0$) in the limit $\epsilon\to 0$ coincide with the generalized topological recursion differentials by Theorem~\ref{th:analytic} (to apply this theorem we note that the poles of $y_\epsilon$ in $U$ are~\regular{} in our assumptions) and with the Chekhov--Norbury topological recursion differentials by a very special case of~\cite[Theorem 5.8]{limits}. This implies the statement of the theorem. 
\end{proof}

\subsection{Comparison with the Bouchard--Eynard recursion}

A version of topological recursion due to Bouchard and Eynard (BE) is designed to work for the multiple critical points of $x$~\cite{BE}. It provides a different expression for the differential $\tilde\omega^{(g)}_{n+1}(z,z_{\set n})$ defining the principal part of the pole of constructed differential $\omega^{(g)}_{n+1}(z,z_{\set n})$. We would not recall the definition in this paper, but rather refer to the original paper of Bouchard--Eynard~\cite{BE} as well as to~\cite{BBCCN} and~\cite{limits}, where a further generalization of this version of topological recursion and its behavior in the families is discussed in detail. To be definite about the definition that we use we refer to~\cite[Definition 2.11]{limits}. 

It is important for this definition that~$x$ and~$y$ are univalued in a neighborhood of the considered point. In the most general setup one can allow $y$ to have poles at the critical points of $x$, but then some extra conditions are required to have the resulting multi-differentials symmetric in all variable~\cite{BBCCN}. In terms of~\cite{limits} these are the extra requirements of the so-called \emph{local admissibility}, see~\cite[Definition 2.5]{limits}. However, no extra conditions are needed if at each critical point of $x$ either $y$ has a simple pole or $y$ is regular and $dy$ is non-vanishing. 

\subsubsection{Compatibility in special cases}

\begin{theorem} If at each critical point of $x$ either $y$ has a simple pole or $y$ is regular and $dy$ is non-vanishing, then the differentials of the BE topological recursion coincide with the differentials of the generalized topological recursion, where the set $\cP$ is chosen to be the set of the critical points of $x$.
\end{theorem}

\begin{proof}  Consider a variation of initial data of the recursion parameterized by a small parameter $\epsilon$ such that for $\epsilon\not=0$ each multiple critical point $q$ of $x$ splits into a number of simple critical points $q_1,\dots,q_r$ in a neighborhood of $q$, where $y$ is unramified, that is, either $y$ or $y^{-1}$ is a local coordinate. 
	%In order to simplify the argument assume also that $q=q_1$, that is, if $y$ had a pole at a critical point of $x$, this poles remains at a critical point of $x$ for $\epsilon \ne 0$. 
	
	In this situation, it is proved 
%that if $y$ is unramified, that is,  either $dy\ne0$ or $y$ has a simple pole at the considered point, then 
in~\cite{limits} that the BE topological recursion is compatible with the limit $\epsilon\to 0$ (it is a special case of Theorem 5.8 in op.~cit.). By Theorem~\ref{th:analytic} the same property holds for the generalized topological recursion. Moreover, we have specified the condition such that for $\epsilon\not=0$ we are in the setup of Theorem~\ref{thm:CN}, that is, for $\epsilon \ne 0$ the two versions of topological recursion do coincide. 

Combining these observations, we conclude that the BE topological recursion agrees with the generalized topological recursion in this setting.
\end{proof}

\subsubsection{Non-compatibility in general}

However, the two topological recursions do not agree for more complicated degeneracies of the spectral curve data. The BE topological recursion does not behave properly if $y$ is ramified: it might happen that $y$ attains equal values on two distinct preimages of~$x$, and this might produce undesired singularities. 

The generalized topological recursion of this paper is free of this disadvantage (in particular, let us stress that there are no requirements of \emph{local admissibility} needed for the BE version of topological recursion) and it is compatible with limits in a much more general setting than the BE one.

Nevertheless, as the following two examples show, the BE topological recursion in certain cases can give an interesting answer, different from the one given by generalized topological recursion, so sometimes the BE topological recursion in a degenerate case may be preferable.

\begin{example}
	Consider the rational spectral curve
	\begin{equation}\label{eq:rscurve}
		x=z^r,\qquad y=z^s,\qquad r\ge2,\quad s<0\text{ or }s=1,
	\end{equation}
	of Example~\ref{ex:rsfirst}. One can deform this spectral curve in such a way that its \nonregular{} point $z=0$ splits into a number of nondegenerate critical points of $x$. For instance, one can take a deformation of the form $x=x(z)$, $y=z^s$, where $x(z)$ is a polynomial of degree $r$ or a rational function having a pole of order $r$ at infinity. If $x(z)$ is generic, then the CEO topological recursion for the deformed curve can be applied. Its differentials can be given by an explicit formula, see~\eqref{eq:dual-to-trivial} below. These differentials have a limit as $x$ degenerates to $z^r$, the limiting differentials are the differentials of the generalized topological recursion for the spectral curve~\eqref{eq:rscurve} with the \nonregular{} point $z=0$ chosen as an \xvital{} one, and these differentials are given by the same Eq.~\eqref{eq:dual-to-trivial} for the limiting spectral curve. In particular, for $s<-r$ the limiting differentials are trivial.
	
	Consider, however, a slightly different situation. Assume that the exponents $r,s$ satisfy the relation
	\begin{equation} \label{eq:r-s-comparison}
		r=m(r+s)\pm1
	\end{equation}
	for some positive integer~$m$. In that case, consider the deformation of the spectral curve~\eqref{eq:rscurve} of the following special form
	\begin{equation}\label{eq:rscurve-deformed}
		x(z)=z^{\pm1}\phi(\Theta(z)),\qquad y=\frac{\Theta(z)}{x(z)},
	\end{equation}
	where $\phi(\theta)$ and $\Theta(z)$ are polynomials of degrees $m$ and $r+s$, respectively, with the unit top coefficients. 
	
	If these polynomials are generic, then $x(z)$ has $m\,(r+s)$ simple critical points, and the CEO topological recursion can be applied and is equivalent to the generalized topological recursion if the set $\cP$ is chosen to be the set of the critical points of $x$. If $\phi$ and $\Theta$ degenerate to monomials, then the spectral curve~\eqref{eq:rscurve-deformed} degenerates to~\eqref{eq:rscurve}. It proves out, however, that the topological recursion differentials for the spectral curve~\eqref{eq:rscurve-deformed} \emph{do not converge} to the differentials of the generalized topological recursion for the spectral curve~\eqref{eq:rscurve}. The reason is that not only the critical points of~$x$ but also other \nonregular{} points collapse together at the point $z=0$ of the limiting curve, namely, the ones which are the critical points of the function $y$.
	Since these \nonregular{} points are not treated as \xvital{} for the generalized topological recursion for generic $\phi(\theta)$ and $\Theta(z)$, Theorem~\ref{th:analytic} of compatibility with the limits cannot be applied.
	
	However, the CEO topological recursion differentials of the family~\eqref{eq:rscurve-deformed} do have a definite limit as $\phi$ and $\Theta$ degenerate to monomials, and the limiting differentials are govern by the BE recursion.
	
	\begin{proposition} 
		If $\phi$ and $\Theta$ are generic (so that all critical points of x are simple and  x has no double zeroes), then the CEO differentials of the spectral curve~\eqref{eq:rscurve-deformed} are given by the following explicit formula:
		\begin{align}\label{eq:OS}
			\omega_n & =
			\prod_{i=1}^n\biggl(
			\sum_{j,k\ge0}\bigl(-d_i\tfrac{x_i}{dx_i}\bigr)^j
			\Bigl(\restr{\theta}{\Theta_i}[v^j]e^{-v\log\phi(\theta)}\partial_\theta^k e^{v\frac{\cS(v\hbar\partial_\theta)}{\cS(\hbar\partial_\theta)}\log\phi(\theta)}\Bigr)[u_i^k]
			\biggr)
			\\ \notag & \qquad
			\prod_{i=1}^n\Bigl(e^{u_i(\cS(u_i\hbar\,z_i\partial_{z_i})-1)\Theta_i}dz_i\Bigr) \cdot
			(-1)^{n-1}\!\!\sum_{\sigma\in\{{\rm cycl}(n)\}}\prod_{i=1}^n\frac{1}{e^{\frac{\pm u_{i}\hbar}{2}}z_i-e^{\frac{\mp u_{\sigma(i)}\hbar}{2}}
				z_{\sigma(i)}}
			\\ \notag & \quad-
			\delta_{n,1}\hbar^{-1}\sum_{j=1}^\infty \bigl(-d_1\tfrac{x_1}{dx_1}\bigr)^{j-1} [v^j]
			\left(\restr{\theta}{\Theta_1}
			e^{v\bigl(\frac{\cS(v\hbar\partial_{\theta})}{\cS(\hbar\partial_\theta)}-1\bigr)\log\phi(\theta)}
			d\theta
			\right),
		\end{align}
		where ${\rm cycl}(n)\subset S(n)$ is the set on $n$-cycles in the permutation group and
		 the sign $\pm$ is the same as in~\eqref{eq:rscurve-deformed}. If $\phi$ and $\Theta$ are not necessarily generic, then the same formula describes the differentials of the BE topological recursion. In particular, these differentials have a limit as $\phi$ and $\Theta$ degenerate to monomials; the limiting differentials are given by the same formula~\eqref{eq:OS} and coincide with the differentials of BE topological recursion for the limiting spectral curve~\eqref{eq:rscurve}.
	\end{proposition}

	\begin{proof}
		The first assertion follows from the theory of symplectic duality in topological recursion, see~\cite{ABDKS-log-sympl}, Eq.~(51), Theorem 6.5, etc. For the other assertions we notice that the family~\eqref{eq:rscurve-deformed} satisfies the conditions sufficient for compatibility of BE recursion with the deformations proved in~\cite[Theorem 5.8]{limits}. The latter check essentially coincides with the one performed in~\cite[Proposition 6.10]{limits}; in particular, the key property that $\omega^{(0)}_1$ separates the fibers follows directly from~\eqref{eq:rscurve-deformed}, which features and explains the role of the property~\eqref{eq:r-s-comparison}. The latter property is precisely the \emph{local admissibility} condition for the BE topological recursion in this case.
	\end{proof}
	
	%\ser{It worth to note that both the differentials of general TR and those of BE recursion for the spectral curve~\eqref{eq:rscurve} as well as for all its deforations considered in this example possess KP integrability property, see Sect.~\ref{sec:KP} for details.}
	
	To conclude, we can associate two different topological recursions with the spectral curve~\eqref{eq:rscurve}. Both of them can be realized as the limit of the CEO topological recursion for suitable deformations of the spectral curve data. For one of these deformations only critical points of~$dx$ collapse together, and this corresponds to the generalized topological recursion of the present paper. In another deformation all critical points of both $x$ and~$y$ collapse together at a single limiting \nonregular{} point, and this corresponds to the BE topological recursion. Both recursions lead to the same differentials in the case $s=\pm1$. But for $s<-1$ the corresponding recursions lead to different results. Notice that the BE recursion is well defined only if the divisibility condition $r\equiv \pm1\pmod {r+s}$ implied by~\eqref{eq:rscurve-deformed} is satisfied: otherwise the differentials defined by the BE recursive formula are not symmetric, see~\cite{BBCCN}.
\end{example}

\begin{example}[{B}ousquet-{M}\'{e}lou--{S}chaeffer numbers]
	Let $\phi(z)$ be a polynomial of degree $m>1$ with $\phi(0)=1$. We call $\phi$ \emph{generic} if the polynomials $\phi(z)$ and $G(z)=\phi(z)-z\,\phi'(z)$ each have $m$ distinct simple roots. Consider the following three rational spectral curves
	\begin{align}
		\label{eq:BMScurve1} x&=\log z-\log\phi(z),&y&=z;
		\\\label{eq:BMScurve2} x&=\frac{z}{\phi(z)},&y&=\phi(z);
		\\\label{eq:BMScurve3} x&=\frac{\phi(z)}{z},&y&=-\frac{z^2}{\phi(z)}.
	\end{align}
	If $\phi$ is generic, then all three curves have the same set of zeroes of $dx$ (which are roots of $G(z)$), the same deck transformations at these points, and the same initial $(0,1)$ differentials
	\begin{equation}
		\omega^{(0)}_1=y\,dx=z\,d\log\Bigl(\frac{z}{\phi(z)}\Bigr)
		=\phi(z)\,d\frac{z}{\phi(z)}
		=-\frac{z^2}{\phi(z)}\,d\frac{\phi(z)}{z}=\frac{G(z)}{\phi(z)}dz.
	\end{equation}
	It follows that the original CEO topological recursions for all three spectral curves are the same. We call the limits of the TR differentials for these spectral curves as $\phi(z)\to (1+z)^m$ the limiting {B}ousquet-{M}\'{e}lou--{S}chaeffer differentials, since their expansions are known to give the so-called  {B}ousquet-{M}\'{e}lou--{S}chaeffer numbers~\cite{BDS-BMSnumbers}. This limit does exist, and we provide a closed formula for the limiting differentials in~\cite{ABDKS-logTR-xy}. Let us discuss what kind of topological recursion the limiting differentials do obey.
	
	With the limit, one of the roots of $G$ tends to $z=\frac1{m-1}$, and $m-1$ other roots collapse together at $z=-1$. The point $z=\frac1{m-1}$ of the limiting spectral curves is a nondegenerate critical point of~$x$, and the contribution of this critical point to the topological recursion is a standard one, for all three versions of the spectral curve data. Let us study in more details what happens at the point $z=-1$.
	
	The limiting {B}ousquet-{M}\'{e}lou--{S}chaeffer differentials have a pole at $z=-1$. However, its principal part disagrees with the contribution suggested by the generalized topological recursion associated with either of the spectral curves~\eqref{eq:BMScurve1}--\eqref{eq:BMScurve3} in the case $\phi(z)=(1+z)^m$. Moreover, for the spectral curves~\eqref{eq:BMScurve2} and~\eqref{eq:BMScurve3} the point $z=-1$ is even not \nonregular{} in the sense of Definition~\ref{def:singpoint} and does not contribute to the generalized topological recursion at all. It means that the generalized topological recursion is not compatible with the limit in all three cases. The reason is that along with the critical points of~$x$, we have other kinds of \nonregular{} points collapsing at the same point $z=-1$. For the spectral curve~\eqref{eq:BMScurve1} these are logarithmic singularities of~$x$, and for the spectral curves~\eqref{eq:BMScurve2}--\eqref{eq:BMScurve3} these are the critical points of~$y$. For the compatibility with the limit, these additional \nonregular{} points should be included to the recursion. But since in the definition of the generalized topological recursion they are not included, the compatibility with the limit breaks.
	
	As for the BE topological recursion, it is not applicable for the spectral curve~\eqref{eq:BMScurve1}, since the $x$ function is not univalued in this parametrization. On the other hand, the BE topological recursion for the spectral curves~\eqref{eq:BMScurve2} and~\eqref{eq:BMScurve3} in the case $\phi(z)=(1+z)^m$ are equivalent to one another, and the differentials produced by this recursion are exactly the desired limiting differentials giving the {B}ousquet-{M}\'{e}lou--{S}chaeffer numbers, see~\cite{BDS-BMSnumbers}, with further comments in~\cite{limits} and~\cite{BDKS2}.	
\end{example}

%%%%%%%%%%%%%%%%%%%%%%%
%%%%%%%%%%%%%%%%%%%%%%%
%%%%%%%%%%%%%%%%%%%%%%%
%%%%%%%%%%%%%%%%%%%%%%%
%%%%%%%%%%%%%%%%%%%%%%%

\section{KP integrability} 
\label{sec:KP}

The goal of this section is to extend the main result of~\cite{alexandrov2024topologicalrecursionrationalspectral} to the case of generalized topological recursion. To this end, we have first to recall the definition of KP integrability as a property of a system of differentials proposed in~\cite{ABDKS3}.

\subsection{Definition of KP integrability for a system of differentials}

Let $\Sigma$ be a smooth Riemann surface and let $\{\omega^{(g)}_n\}$ be system of $n$-differentials on $\Sigma^n$ defined for all $g\ge0$, $n\ge1$. Assume that $\omega^{(g)}_n$'s are symmetric and meromorphic with no poles on the diagonals for $(g,n)\ne(0,2)$, and $\omega^{(0)}_2$ is also symmetric and meromorphic but it has a second order pole on the diagonal with bi-residue~$1$. For the moment there are no further assumptions. 

\begin{definition}
	A point $o\in\Sigma$ is called \emph{regular} for the system of differentials $\{\omega^{(g)}_n\}$ if in some (and hence in every) local coordinate~$\lcoord$ near $o$ we have that $\omega^{(g)}_n-\delta_{(g,n),(0,2)}\frac{d\lcoord_1d\lcoord_2}{(\lcoord_1-\lcoord_2)^2}$ is regular at $(o,\dots,o)\in\Sigma^n$ for all $(g,n)\in\Z_{\geq 0}\times\Z_{>0}$, $(g,n)\not=(0,1)$.
\end{definition}

Let $o\in\Sigma$ be a regular point and $\lcoord$ be an arbitrary local coordinate at $o$. Consider the expansions
\begin{equation}\label{eq:omega-expansion}
	\sum_{\substack{g\geq 0\\ 2g-2+n\geq 0}}\hbar^{2g-2+n}\omega^{(g)}_n=\delta_{n,2}\tfrac{d\lcoord_1d\lcoord_2}{(\lcoord_1-\lcoord_2)^2}+\sum_{g=0}^\infty \hbar^{2g-2+n}
	\sum_{k_1,\dots,k_n=1}^{\infty} f^{(g)}_{k_1,\dots,k_n}\prod_{i=1}^n k_i \lcoord_i^{k_i-1} d\lcoord_i.
\end{equation}
Let
\begin{equation}\label{eq:F-expansion}
	Z_{o,\lcoord}\coloneqq \exp \sum_{\substack{g\geq 0, n\geq 1 \\ 2g-2+n \geq 0}}\frac{\hbar^{2g-2+n}}{n!} 
	\sum_{k_1,\dots,k_n=1}^{\infty} f^{(g)}_{k_1,\dots,k_n}\prod_{i=1}^n {p}_{k_1}\dots {p}_{k_n},
\end{equation}
where ${p}_i$, $i=1,2,\dots$, are formal variables.

\begin{definition} \label{def:KPdiff} A system of differentials $\{\omega^{(g)}_n\}$ satisfying the assumptions as above is called \emph{KP integrable} if the associated partition function $Z_{o,\lcoord}$ is a KP tau function for at least some choice of a regular point $o\in\Sigma$ and a local coordinate $\lcoord$.
\end{definition}

This definition is justified by \cite[Theorem 2.3]{ABDKS3} which states that the thus defined KP integrability is an internal property of the system of differentials and does not depend on the choice of  a regular point $o$ and a coordinate $\lcoord$ at $o$. It follows that the KP integrability can be regarded as an internal property of a system of differentials.
%This treatment of KP integrability has a quite definite formulation. 

This KP integrability property can equivalently be reformulated via the so-called determinantal formulas. 
Introduce the \emph{Baker-Akhiezer kernel} $\bK$ for a given system of differentials by the following formula
\begin{equation} \label{eq:BA-kernel}
	\begin{aligned}
		\omega^{(0),{\rm sing}}_2(z_1,z_2)&=\frac{dz_1dz_2}{(z_1-z_2)^2},
		\\\bK(z_1,z_2)&=
		\tfrac{\sqrt{dz_1dz_2}}{z_1-z_2}
		\exp \left({\sum\limits_{(g,n)\ne(0,1)}\frac{\hbar^{2g-2+n}}{n!}
			\int\limits_{z_2}^{z_1}\dots\int\limits_{z_2}^{z_1}(\omega^{(g)}_n-\delta_{(g,n),(0,2)}\omega^{(0),{\rm sing}}_2)}\right),
	\end{aligned}
\end{equation}
where $z$ is an arbitrary local coordinate or just a meromorphic function. By definition, $\bK$ is a half-differential on~$\Sigma^2$, and a reason to introduce it this way is that it is actually independent of a choice of~$z$. It is well-defined in a vicinity of the diagonal, but in certain cases it extends globally to $\Sigma^2$. Then we have: %the KP integrability of a system of differentials can be expressed via the so-called determinantal formulas using the following statement:

\begin{theorem}[\cite{ABDKS3}]
	A system of differentials $\{\omega^{(g)}_n\}$ is KP integrable if and only if the following determinantal identities hold true
	\begin{equation}\label{eq:detc}
		\begin{aligned}
			\omega_1(z_1)&=\lim\limits_{z_2\to z_1}\Bigl(\bK(z_1,z_2)-\sqrt{\omega^{(0),{\rm sing}}_2(z_1,z_2)}\Bigr),
			\\\omega_n(z_{\set n})&=(-1)^{n-1}\sum_{\sigma\in{\rm cycl}(n)}\prod_{i=1}^n\bK(z_i,z_{\sigma(i)}),\quad n\ge2,
		\end{aligned}
	\end{equation}
	where the summation goes over the set of $n$-cycles in the permutation group $S(n)$.
\end{theorem}

\subsection{KP integrability for generalized topological recursion}

Consider the initial data $(\Sigma,dx,dy,B,\cP)$ for the generalized topological recursion such that
\begin{itemize}
	\item $\Sigma=\C P^1$;
	\item $B(z_1,z_2)=\frac{dz_1dz_2}{(z_1-z_2)^2}$, where $z$ is a global affine coordinate on $\C P^1$.
\end{itemize}

The generalized topological recursion associates to this input data a system of differentials $\{\omega^{(g)}_n\}$.  The following property, known for the original CEO topological recursion in the case of the rational spectral curve, holds in the setup of generalized topological recursion as well:

\begin{theorem} \label{thm:KP-integrability} The system of differentials  $\{\omega^{(g)}_n\}$ obtained by the generalized topological recursion  for the input data as above is KP
	integrable.
\end{theorem}

\begin{proof} In order to prove this statement, we reduce it to the analogous statement for the original CEO topological recursion proven in~\cite{alexandrov2024topologicalrecursionrationalspectral}. Indeed, the property of a system of differentials to be KP integrable can be described via determinantal formulas  and thus is a closed condition. Hence, if we can realize a system of differentials  $\{\omega^{(g)}_n\}$ as a limit of a system of differentials $\lim_{\epsilon\to 0}\{\omega^{(g)}_n (\epsilon) \}$ that analytically depends on a small parameter $\epsilon$ and is KP integrable for a general $\epsilon\not=0$, we can conclude that  $\{\omega^{(g)}_n\}$ is KP integrable as well. Here we exclude the $(g,n)=(0,1)$ differential from the consideration as it doesn't affect the KP integrability property. 
	
	Now, consider the initial data as above and defined a one-parameter deformation with respect to a small parameter $\epsilon$ such that $dy$ is not deformed and the deformation of $dx$ is chosen as
	\begin{align}
		dx(\epsilon)\coloneqq dx + \epsilon \sum_{q\in \cP} \frac{dz}{(z-q)^{A+1}},
	\end{align}
	where $A\in \Z_{>0}$ is strictly bigger than the order of zeroes of $dy$ at the points in $\cP$ (the poles of $dy$ are then counted as zeroes of non-positive order). By Theorem~\ref{th:analytic}, for small enough $\epsilon$, the generalized topological recursion applied to $(\C P^1, dx(\epsilon), dy, B(z_1,z_2)=\frac{dz_1dz_2}{(z_1-z_2)^2}, \cP(\epsilon))$, where $\cP(\epsilon)$ is the set of all \nonregular{} points located in the union $U\coloneqq \cup_{q\in \cP} U_q$ of some fixed system of small disks around the points $q\in \cP$, returns a system of differentials $\{\omega^{(g)}_n (\epsilon) \}$ such that indeed $\lim_{\epsilon\to 0}\{\omega^{(g)}_n (\epsilon) \}=\{\omega^{(g)}_n\}$.
	
	Note that for a generic small value of $\epsilon$ the set of \nonregular{} points in $U$ is  the set of simple zeros of $dx(\epsilon)|_{U}$. In particular, if the order of zero of $dx$ at some point $q\in \cP$ is equal to $B$, then $dx(\epsilon)$ has a pole of order $(A+1)$ at $q$ and there are further $(A+B+1)$ simple zeros of $dx(\epsilon)$ close to $q$. These $(A+B+1)$ simple zeros of $dx(\epsilon)$ all belong to $\cP(\epsilon)$. 
	
	The above discussion implies that for $\epsilon\not=0$ the points in $\cP$ are \regular{}, since $A$ is chosen to be big enough. Indeed, in terms of the parameters $r$ and $s$ used in Definition~\ref{def:singpoint} at each point of $q\in\cP$ we have for $\epsilon\not=0$ that $r=-A$ and $s-1<A$, hence $r+s\leq 0$. Thus all points in $\cP(\epsilon)$ for $\epsilon \not=0$ are simple zeroes of $dx(\epsilon)$, where $dy$ is regular and non-vanishing. 
	By Theorem~\ref{th:CEO-compatibility} the differentials of the generalized topological recursion in this case coincide with the differentials of the original CEO topological recursion for the input data given by $\C P^1$, $B$, and the functions $x$ and $y$ chosen as some primitives for $dx|_U$ and $dy|_U$ (and thus defined only locally on $U$). 
	
	\begin{remark}
		This is a relaxed setup for the original CEO topological recursion, cf.~\cite[Section 4]{alexandrov2024topologicalrecursionrationalspectral}. One can think of topological recursion applied to the curve $U\subset \C P^1$ with the bi-differential $\frac{dz_1dz_2}{(z_1-z_2)^2}|_U$, but then we consider the obtained differentials $\omega^{(g)}_n(\epsilon)$ for $2g-2+n\geq 0$ as the ones defined on $(\C P^1)^n$ since the projection formula with the globally defined $B$ extends them beyond $U$. 	
	\end{remark}

	We see that the setup in which we produce the system of differentials $\{\omega^{(g)}_n (\epsilon) \}$ for a generic non-zero value of $\epsilon$ satisfies the requirements of~\cite[Theorem 4.1]{alexandrov2024topologicalrecursionrationalspectral}, in which the KP integrability for the system of differentials of the original CEO topological recursion is proven in \emph{op.~cit.} Thus we constructed a family of differentials $\{\omega^{(g)}_n (\epsilon) \}$ which is analytic in $\epsilon$ at $\epsilon=0$, KP integrable for a generic non-zero value of $\epsilon$, and $\lim_{\epsilon\to 0}\{\omega^{(g)}_n (\epsilon) \}=\{\omega^{(g)}_n\}$. This completes the proof of the theorem.
\end{proof}

%%%%%%%%%%%%%%%%%%%%%%%
%%%%%%%%%%%%%%%%%%%%%%%
%%%%%%%%%%%%%%%%%%%%%%%
%%%%%%%%%%%%%%%%%%%%%%%
%%%%%%%%%%%%%%%%%%%%%%%
%%%%%%%%%%%%%%%%%%%%%%%
%%%%%%%%%%%%%%%%%%%%%%%
%%%%%%%%%%%%%%%%%%%%%%%
%%%%%%%%%%%%%%%%%%%%%%%

\section{Examples}\label{sec:rs-examples}

\subsection{The case of trivial dual TR}
We start this section with the following observation. Assume that $\Sigma=\C P^1$ with an affine coordinate~$z$. Assume also that \emph{all} \nonregular{} points of the initial data of topological recursion are set to be \xvital. Then the $x-y$ dual topological recursion is trivial. This leads to an explicit closed formula for the TR differentials that can be represented in two different equivalent forms. For the first one we denote
\begin{equation}
\begin{aligned}
\hat z(z,v)&=e^{\frac {v \hbar}{2}\partial_{y}}z=z+\tfrac{v}{2 y'}\hbar-\tfrac{v^2 y''}{8 (y')^3}\hbar^2-\tfrac{v^3 (y' y^{(3)}-3 (y'')^2)}{48 (y')^5}\hbar^3+\dots,
\\\hat z_i^\pm&=\hat z(z_i,\pm v_i).
\end{aligned}
\end{equation}
With this notation, the formula reads
\begin{equation}\label{eq:dual-to-trivial}
\begin{aligned}
\bW^{\vee}_n(z_{\set n},v_{\set n})&=
\prod_{i=1}^n\Bigl( e^{v_i(\cS(v_i\hbar\partial_{y_i})-1)x_i}\sqrt{\tfrac{d\hat z^+_i}{d z_i}\tfrac{d\hat z^-_i}{dz_i}}\,dz_i\Bigr)
~(-1)^{n-1}\!\!\sum_{\sigma \in {\rm cycl(n)}}\prod_{i=1}^n\frac1{\hat z_i^+-\hat z_{\sigma(i)}^-},
\\\omega_n(z_{\set n})&=(-1)^n\Bigl(\prod_{i=1}^n\sum_{k=0}^\infty\bigl(-d_i\tfrac{1}{dx_i}\bigr)^k[v_i^k]\Bigr)
\;\bW^{\vee}_n(z_{\set n},v_{\set n}),
\end{aligned}
\end{equation}
where $x'_i=x'(z_i)$, $y'_i=y'(z_i)$.
For example, for small $(g,n)$ we get
\begin{align}
\omega^{(0)}_3(z_1,z_2,z_3)&=d_1d_2d_3\sum_{i=1}^3\frac{1}{x'_iy'_i}\prod_{j\ne i}\frac{1}{z_i-z_j},
\\\omega^{(1)}_1(z)&=\frac1{24}d\left(\frac{3 (y'')^2}{x' (y')^3}-\frac{2 y^{(3)}}{x' (y')^2}
  +\partial_{x}  \Bigl(\frac{y''}{(y')^2}-\frac{x''}{x' y'}\Bigr)\right).
\end{align}

The second closed formula for the differentials~$\omega^{(g)}_n$ is the determinantal one %~\eqref{eq:detc} 
expressing KP integrability of these differentials~\cite{ABDKS3}. To this end, we have proved that the Baker–Akhiezer kernel~\eqref{eq:BA-kernel} can be computed directly in this case and represented as the
following formal double Gaussian integral:
\begin{theorem}[\cite{ABDKS3}] We have:
\begin{multline}\label{BAK}
	\frac{\bK(z_1,z_2)}{\sqrt{dz_1dz_2}}=\frac{\sqrt{-x'_1x'_2}}{2\pi\hbar}\iint\frac{\sqrt{y'(\xi_1)y'(\xi_2)}}{\xi_1-\xi_2}
	d\xi_1d\xi_2\times\\
	e^{\hbar^{-1}((y(\xi_1)-y_1)x_1-\int_{z_1}^{\xi_1}x\,dy)}
	e^{-\hbar^{-1}((y(\xi_2)-y_2)x_2-\int_{z_2}^{\xi_2}x\,dy)}.
\end{multline}	
\end{theorem}
For an explanation %of the terminology and 
how this formula is treated and used see~\cite{ABDKS3}. Here we just notice that this formula represents $\bK$ as a power series in~$\hbar$, and the coefficient of each monomial in~$\hbar$ is a polynomial combination of $1/(z_1-z_2)$, $1/x'_i$, $1/y'_i$, $i=1,2$, and the derivatives $x^{(k)}_i$, $y^{(k)}_i$, $i=1,2$ of different orders $k\ge2$:
\begin{multline}
\frac{\bK(z_1,z_2)}{\sqrt{dz_1dz_2}}=\frac{1}{z_1-z_2}+\biggl(
\frac{\tfrac{1}{x'_1 y'_1}-\tfrac{1}{x'_2 y'_2}}{(z_1-z_2)^3}
+\frac{\tfrac{x''_1}{(x'_1)^2 y'_1}+\tfrac{y''_1}{x'_1(y'_1)^2}+\tfrac{x''_2}{(x'_2)^2 y'_2}+\tfrac{y''_2}{x'_2(y'_2)^2}}{2(z_1-z_2)^2}
\\+\frac{1}{24(z_1-z_2)}\bigl(-\tfrac{3 x_1^{(3)}}{(x_1^{\prime })^2 y_1^{\prime }}-\tfrac{3 y_1^{(3)}}{x_1^{\prime } (y_1^{\prime })^2}+\tfrac{5 (x_1^{\prime\prime })^2}{(x_1^{\prime })^3 y_1^{\prime }}+\tfrac{5 x_1^{\prime\prime } y_1^{\prime\prime }}{(x_1^{\prime })^2 (y_1^{\prime })^2}+\tfrac{5 (y_1^{\prime\prime })^2}{x_1^{\prime } (y_1^{\prime })^3}
\\+\tfrac{3 y_2^{(3)}}{x_2^{\prime } (y_2^{\prime })^2}+\tfrac{3 x_2^{(3)}}{(x_2^{\prime })^2 y_2^{\prime }}-\tfrac{5 (x_2^{\prime\prime })^2}{(x_2^{\prime })^3 y_2^{\prime }}-\tfrac{5 x_2^{\prime\prime } y_2^{\prime\prime }}{(x_2^{\prime })^2 (y_2^{\prime })^2}-\tfrac{5 (y_2^{\prime\prime })^2}{x_2^{\prime } (y_2^{\prime })^3}\bigr)\biggr)\hbar+\dots
\end{multline}

We list below some cases when these computations lead to a nontrivial answer even for CEO topological recursion (cf. \cite{ABDKS-logTR-xy}, Sect.3).

\subsection{The case \texorpdfstring{$y=z^s$}{y=z s}}
Consider the following spectral curve
\begin{equation}\label{eq:y=z^s}
x=x(z),\qquad y=z^s,
\end{equation}
where $x(z)$ is rational. For definiteness, we assume that $s<0$ (otherwise we can switch the affine coordinate $z$ to $1/z$). The \nonregular{} points in a sense of Definition~\ref{def:singpoint} for this spectral curve data are all finite critical points of~$x$ and also possibly infinity. In fact, the point $z=\infty$ is \regular{} if $x$ has a pole of order $r\ge-s$ at this point. We conclude:

\begin{proposition}
Let $s<0$ and $x(z)$ be a polynomial of degree $r\ge-s$ or a rational function with a pole of order $r\ge-s$ at $z=\infty$. Assume also that all critical points of $x(z)$ are simple and distinct from zero. Then the differentials of CEO topological recursion for the spectral curve~\eqref{eq:y=z^s} are given explicitly by~\eqref{eq:dual-to-trivial}.

Moreover, assume that $x(z)$ depends on an additional parameter~$\tau$ such that $x\bigm|_{\tau=0}=z^r$, e.g. $x(z)=z^r+\tau z$. Then the TR differentials of this family of the spectral curve data have a limit as $\tau\to 0$ and the limiting differentials are differentials of the (generalized) topological recursion with the spectral curve $x=z^r$, $y=z^s$, and are given by the same formula~\eqref{eq:dual-to-trivial} for the limiting function $x=z^r$.
\end{proposition}

Let us consider now a slightly more general situation. Assume that the function $x(z)$ of the spectral curve~\eqref{eq:y=z^s} admits logarithmic singularities. Namely, we only require that the differential $dx$ is rational. In this case along with the zeroes of~$dx$ we have extra \nonregular{} points of the spectral curve data, namely, the ones that are the simple poles of~$dx$. Since they do not contribute to the original CEO topological recursion, we consider these points as \yvital{} in the setup of generalized topological recursion. The $x-y$ dual differentials for this initial data of recursion are not trivial any more but they can easily be computed explicitly, see~\cite{ABDKS-logTR-xy} and~\cite{hock2023xy}. Namely, we have $\omega^{\vee,(g)}_n=0$ for $n\ge2$, $(g,n)\ne(0,2)$, and for $n=1$ we have
\begin{equation}
\omega^{\vee}_1=\hbar^{-1}(\hat x-x)\,dy,
\end{equation}
where the function $\hat x=\hat x(z,\hbar)$ is of the form $\hat x=x+O(\hbar^2)$ and is defined as follows. Let $\{a_1,\dots,a_m\}$ be the list of all \emph{simple nonzero} poles of $dx$, and let $\alpha_1^{-1},\dots,\alpha_m^{-1}$ be the corresponding residues. Then, we set
\begin{equation}
\hat x(z)=x(z)+\sum_{i=1}^m\left(\frac{1}{ \cS(\alpha_i\hbar\partial_y)}-1\right)\frac{\log(z-a_i)}{\alpha_i}.
\end{equation}
Note that we allow $dx$ to have other poles of order greater than one with possibly nonzero residues as well as a simple pole at $z=0$; they do not contribute to $\hat x-x$.
As a corollary, we have:

\begin{proposition}
Let $s<0$ and $dx(z)$ be a rational $1$-form with a pole of order $r+1$ at $z=\infty$ for some $r\ge-s$. Assume also that all zeroes of $x(z)$ are simple and distinct from $z=0$. Then the differentials of CEO topological recursion for the spectral curve~\eqref{eq:y=z^s} are given explicitly by~\eqref{eq:dual-to-trivial} with the following modification of the dual extended differentials:
\begin{equation}
\bW^{\vee}_n(z_{\set n},v_{\set n})=
\prod_{i=1}^n\Bigl( e^{v_i(\cS(v_i\hbar\partial_{y_i})\hat x_i-x_i)}\sqrt{\tfrac{dz^+_i}{dz_i}\tfrac{dz^-_i}{dz_i}}\,dz_i\Bigr)
~(-1)^{n-1}\!\!\sum_{\sigma\in {\rm cycl(n)}}\prod_{i=1}^n\frac1{z_i^+-z_{\sigma(i)}^-}
\end{equation}
\end{proposition}

\subsection{Topological recursion for \texorpdfstring{$(r,s)$}{r,s} spectral curves}\label{S6.3}
In this section, we provide the results of computations of topological recursion with the spectral curve
\begin{equation}\label{eq:rs-spectralcurve}
dx=z^{r-1}dz,\qquad dy=z^{s-1}dz
\end{equation}
for different pairs of integers $(r,s)$. If $r\ne0,s\ne0$, then we have
\begin{equation}
x=z^r/r,\qquad y=z^s/s,
\end{equation}
respectively. The numerical factors for these monomials are not important, it is a matter of normalization. What is essential is that we include the cases $r=0$ or $s=0$ as well when the corresponding functions attain logarithmic singularities.

If $r+s=0$ or $(r,s)=(1,1)$, then the set of \nonregular{} points is empty and the topological recursion is trivial. The case $r+s<0$ can be reduced to the case $r+s>0$ by changing the coordinate $z$ to its inverse. Thus, we can assume without loss of generality that $r+s>0$. In this case $z=0$ is the unique \nonregular{} point of the spectral curve data. We have two options for a choice of the initial data of recursion: this point is either \xvital{} or \yvital. In the first case, the dual TR is trivial and we can make use of Eqs.~\eqref{eq:dual-to-trivial} for the computation of TR differentials. The second case is reduced to the first one by the swap of $x$ with~$y$ and $r$ with~$s$, respectively. Thus, by the topological recursion for the spectral curve~\eqref{eq:rs-spectralcurve} we mean a nontrivial one, with the \nonregular{} point $z=0$ considered as \xvital{} (independently of the signs of~$r$ and~$s$).

The differentials $\omega^{(g)}_n$ of this topological recursion are Laurent polynomial in $z$-variables and their expansion~\eqref{eq:omega-expansion} can be written as 
\begin{equation}
\omega^{(g)}_n=\sum_{k_1,\dots,k_n\ge1}f^{(g)}_{k_1,\dots,k_n}\prod_{i=1}^n d(z_i^{-k_i}),\quad 2g-2+n>0,
\end{equation}
with finitely many nonzero terms. It is convenient to collect the coefficients of these expansions into the corresponding potential
\begin{equation}
F(p_1,p_2,\dots;\hbar)=\sum_{2g-2+n>0}\frac{\hbar^{2g-2+n}}{n!}\sum_{k_1,\dots,k_n\ge1}f^{(g)}_{k_1,\dots,k_n}p_{k_1}\dots p_{k_n}. 
\end{equation}
This potential is a solution of the KP hierarchy for any $(r,s)$, as it follows from Theorem~\ref{thm:KP-integrability}, or, in this particular case, directly from~\cite{ABDKS-logTR-xy}. %Theorem~\ref{th:KP}. 
The result of its computation for some pairs $(r,s)$ is presented in the following table.
{\small
\renewcommand{\arraystretch}{1.5}
\begin{longtable}{|l|l|}
\hline
$\quad(x,y)$   &  \qquad\qquad\qquad $F$\\\hline
$\left(\dfrac{z^{-2}}{2},\dfrac{z^3}{3}\right)$ &$\frac{7}{24} p_1 \hbar -\frac{7}{16} p_1^2 \hbar ^2+(\frac{7 }{8}p_1^3+\frac{455 }{1152}p_3) \hbar ^3
  -(\frac{63}{32} p_1^4+\frac{455}{128} p_3 p_1) \hbar ^4
$\\&\qquad $
+(\frac{189}{40} p_1^5+\frac{1365}{64} p_3 p_1^2+\frac{19019}{3072} p_5) \hbar ^5
  +O(\hbar^6)$\\
$\left(\dfrac{z^{-2}}{2},\dfrac{z^4}{4}\right)$ & $\frac{3 }{8} p_2 \hbar+(\frac{1}{8}p_1^4-\frac{5 }{4}p_3 p_1-\frac{3 }{4}p_2^2) \hbar ^2
  +(+3 p_4 p_1^2+10 p_2 p_3 p_1+2 p_2^3+\frac{65}{8} p_6-p_2 p_1^4) \hbar ^3
$\\&\qquad $
  +(\frac{21}{8} p_3 p_1^5+6 p_2^2 p_1^4-\frac{375}{16} p_3^2 p_1^2-36 p_2 p_4 p_1^2-60 p_2^2 p_3 p_1-\frac{3861}{32} p_7 p_1-6 p_2^4
$\\&\qquad $
  -\frac{207}{8} p_4^2-\frac{2079}{32}p_3 p_5-\frac{195}{2} p_2 p_6) \hbar ^4+O(\hbar ^5)$\\
$\left(z^{-1},\dfrac{z^2}{2}\right)$ & $-\frac{1 }{8}p_1 \hbar-\frac{1}{8} p_1^2 \hbar ^2-(\frac{1}{6}p_1^3+\frac{5}{32} p_3) \hbar ^3-(\frac{1}{4}p_1^4+\frac{15 }{16}p_3 p_1+\frac{15 }{64}p_2^2) \hbar ^4
$\\&\qquad $
-(\frac{2 }{5}p_1^5+\frac{15}{4} p_3 p_1^2+\frac{15}{8} p_2^2 p_1+\frac{189}{64} p_5) \hbar ^5+O(\hbar ^6)$\\
$\left(z^{-1},\dfrac{z^3}{3}\right)$ & $-\frac{1}{6}p_2 \hbar+(\frac{1}{12}p_1^4-\frac{1}{3}p_3 p_1-\frac{1}{4}p_2^2) \hbar ^2
-(2 p_2 p_3 p_1-\frac{1}{2} p_2 p_1^4+\frac{1}{2}p_2^3+\frac{35}{18} p_6) \hbar ^3
+(p_3 p_1^5+\frac{9}{4} p_2^2 p_1^4
$\\&\qquad $
+\frac{14}{3} p_5 p_1^3-2 p_3^2 p_1^2-9 p_2^2 p_3 p_1-15 p_7 p_1-\frac{35 }{2}p_2 p_6-\frac{35}{3} p_3 p_5-\frac{9}{8} p_2^4-\frac{35}{8}p_4^2) \hbar ^4
+O(\hbar ^5)$\\
$\left(\log z,z\right)$ & $-\frac{1}{24}p_1 \hbar -\frac{1}{48} p_1^2 \hbar ^2-(\frac{1}{72}p_1^3+\frac{9 }{640}p_3) \hbar ^3-(\frac{1}{96}p_1^4+\frac{27 }{640}p_3 p_1+\frac{1}{90}p_2^2) \hbar ^4
$\\&\qquad $
-(\frac{1}{120}p_1^5+\frac{27}{320} p_3 p_1^2+\frac{2}{45} p_2^2 p_1+\frac{40625}{580608} p_5) \hbar ^5+O(\hbar ^6)$\\
$\left(\log z,\dfrac{z^2}{2}\right)$ & $-\frac{1}{24}p_2 \hbar +(\frac{1}{24}p_1^4-\frac{1}{24}p_2^2) \hbar ^2+(\frac{1}{6} p_2 p_1^4+\frac{1}{3} p_4 p_1^2-\frac{1}{18}p_2^3-\frac{9}{160} p_6) \hbar ^3 +(\frac{9}{40} p_3 p_1^5+\frac{1}{2} p_2^2 p_1^4
$\\&\qquad $
+\frac{125}{72} p_5 p_1^3+\frac{9}{16} p_3^2 p_1^2+2 p_2 p_4 p_1^2+\frac{343}{180} p_7 p_1-\frac{1}{12}p_2^4+\frac{29}{180} p_4^2-\frac{27}{80} p_2 p_6) \hbar ^4+O(\hbar ^5)$\\
%
%$(\log z,\dfrac{z^3}{3})$ & $+O(\hbar^5)$\\
%
$(z,\log z)$ & $-\frac{1}{24 }p_1\hbar +\frac{7 }{2880}\hbar^3p_3-\frac{31}{40320} \hbar^5p_5+\frac{127}{215040}\hbar^7p_7-\frac{511}{608256}\hbar ^9p_9+O(\hbar ^{11})$\\
$(z,z)$ & $0$\\
$\left(z,\dfrac{z^2}{2}\right)$ & $(\frac{1}{6}p_1^3+\frac{1}{24}p_3) \hbar +(\frac{1}{3} p_3 p_1^3+\frac{1}{4} p_2^2 p_1^2+\frac{1}{2}p_5 p_1+\frac{1}{24}p_3^2+\frac{1}{8}p_2 p_4) \hbar ^2+(\frac{1}{2} p_5 p_1^4+\frac{2}{3} p_3^2 p_1^3
$\\&\qquad $
+p_2 p_4 p_1^3+p_2^2 p_3 p_1^2+\frac{25}{8} p_7 p_1^2+\frac{1}{8} p_2^4 p_1+\frac{3}{4} p_4^2 p_1+2 p_3 p_5 p_1+\frac{5}{2} p_2 p_6 p_1
$\\&\qquad $
+\frac{1}{18}p_3^3+\frac{1}{2} p_2 p_3 p_4+\frac{1}{2} p_2^2 p_5+\frac{91}{48} p_9) \hbar ^3+O(\hbar ^4)$\\
$\left(z,\dfrac{z^3}{3}\right)$ & $(\frac{1}{2} p_2 p_1^2+\frac{1}{12}p_4) \hbar+(\frac{1}{24}p_2^4+p_1 p_3 p_2^2+\frac{3}{2} p_1^2 p_4 p_2+\frac{5}{6} p_6 p_2+\frac{1}{2} p_1^2 p_3^2+\frac{1}{8}p_4^2
$\\&\qquad $
+\frac{2}{3} p_1^3 p_5+\frac{1}{3}p_3 p_5+2 p_1 p_7) \hbar ^2+O(\hbar ^3)$\\
$\left(\dfrac{z^2}{2},z^{-1}\right)$ & $\frac{1}{8}p_1 \hbar+\frac{1}{16} p_1^2 \hbar ^2+(\frac{1}{24}p_1^3+\frac{3}{128} p_3) \hbar ^3+(\frac{1}{32}p_1^4+\frac{9}{128} p_3 p_1) \hbar ^4+(\frac{1}{40}p_1^5+\frac{9}{64} p_3 p_1^2
$\\&\qquad $
+\frac{45}{1024} p_5) \hbar ^5
+(\frac{p_1^6}{48}+\frac{15}{64} p_3 p_1^3+\frac{225 p_5 p_1}{1024}+\frac{63 p_3^2}{1024}) \hbar ^6+O(\hbar ^7)$\\
$\left(\dfrac{z^2}{2},\log z\right)$ & $-\frac{1}{24}p_2 \hbar-(\frac{1}{24}p_1^4+\frac{1}{12}p_3 p_1) \hbar ^2+\frac{7}{720}p_6 \hbar^3
$\\&\qquad $
+(\frac{1}{40} p_3 p_1^5+\frac{1}{12} p_5 p_1^3+\frac{1}{16} p_3^2 p_1^2+\frac{3}{32} p_7 p_1+\frac{23 }{480}p_3 p_5) \hbar ^4
-\frac{31}{2520} p_{10} \hbar^5 +O(\hbar ^6)$\\
$\left(\dfrac{z^2}{2},z\right)$ & $(\frac{1}{6}p_1^3+\frac{1}{24}p_3) \hbar+(\frac{1}{6} p_3 p_1^3+\frac{1}{8}p_5 p_1+\frac{1}{48}p_3^2) \hbar ^2
+(\frac{1}{8} p_5 p_1^4+\frac{1}{6} p_3^2 p_1^3+\frac{5}{16} p_7 p_1^2+\frac{1}{4} p_3 p_5 p_1
$\\&\qquad $
+\frac{1}{72}p_3^3+\frac{35}{384} p_9) \hbar ^3
+(\frac{1}{8} p_7 p_1^5+\frac{3}{8} p_3 p_5 p_1^4+\frac{1}{6} p_3^3 p_1^3+\frac{35}{48} p_9 p_1^3+\frac{3}{8} p_5^2 p_1^2+\frac{15}{16} p_3 p_7 p_1^2
$\\&\qquad $
+\frac{3}{8} p_3^2 p_5 p_1+\frac{105}{128} p_{11} p_1+\frac{1}{96}p_3^4+\frac{29}{128} p_5 p_7+\frac{35}{128} p_3 p_9) \hbar ^4+O(\hbar ^5)$\\
$\left(\dfrac{z^2}{2},\dfrac{z^2}{2}\right)$ & $(\frac{1}{2} p_2 p_1^2+\frac{1}{8}p_4) \hbar +(\frac{1}{2} p_5 p_1^3+\frac{1}{4} p_3^2 p_1^2+p_2 p_4 p_1^2+\frac{1}{2} p_2^2 p_3 p_1+\frac{5}{4} p_7 p_1+\frac{1}{8}p_4^2
$\\&\qquad $
+\frac{1}{4}p_3 p_5+\frac{1}{2}p_2 p_6) \hbar ^2+O(\hbar ^3)$\\
$\left(\dfrac{z^2}{2},\dfrac{z^3}{3}\right)$ & $(\frac{1}{2} p_3 p_1^2+\frac{1}{2} p_2^2 p_1+\frac{5}{24} p_5) \hbar+(\frac{5}{6} p_7 p_1^3+\frac{1}{2} p_4^2 p_1^2+\frac{3}{2} p_3 p_5 p_1^2+2 p_2 p_6 p_1^2+\frac{1}{3} p_3^3 p_1
$\\&\qquad $
+2 p_2 p_3 p_4 p_1+\frac{3}{2} p_2^2 p_5 p_1
+\frac{91}{24} p_9 p_1+\frac{1}{4} p_2^2 p_3^2+\frac{5}{16}p_5^2+\frac{1}{3} p_2^3 p_4+\frac{2}{3} p_4 p_6
$\\&\qquad $
+\frac{25}{24} p_3 p_7+2 p_2 p_8) \hbar ^2+O(\hbar ^3)$\\
$\left(\dfrac{z^3}{3},\dfrac{z^{-2}}{2}\right)$ & $\frac{7}{24} p_1 \hbar +\frac{7}{24} p_1^2 \hbar ^2+\frac{7}{18} p_1^3 \hbar ^3+(\frac{7}{12} p_1^4-\frac{385}{576} p_2^2) \hbar ^4+(\frac{14}{15} p_1^5-\frac{385}{72} p_2^2 p_1-\frac{30107}{10368} p_5) \hbar ^5+O(\hbar ^6)$\\
$\left(\dfrac{z^3}{3},z^{-1}\right)$ & $\frac{1}{6}p_2 \hbar +(\frac{1}{12}p_2^2-\frac{1}{12}p_1^4) \hbar ^2-(\frac{1}{6} p_2 p_1^4+\frac{1}{3} p_4 p_1^2-\frac{1}{18}p_2^3) \hbar ^3
$\\&\qquad $
+(-\frac{1}{4} p_2^2 p_1^4-\frac{5}{9} p_5 p_1^3-p_2 p_4 p_1^2-\frac{7}{9} p_7 p_1+\frac{1}{24}p_2^4-\frac{13}{72} p_4^2) \hbar ^4+O(\hbar ^5)$\\
$(\dfrac{z^3}{3},\log z)$ & $(\frac{1}{6}p_1^3-\frac{1}{24}p_3) \hbar +(-\frac{1}{4} p_1^2 p_2^2-\frac{1}{8}p_4 p_2-\frac{1}{4}p_1 p_5) \hbar ^2+(-\frac{1}{12} p_5 p_1^4-\frac{1}{3} p_2 p_4 p_1^3-\frac{5}{12} p_7 p_1^2
$\\&\qquad $
+\frac{1}{24} p_2^4 p_1-\frac{1}{6} p_4^2 p_1+\frac{1}{12} p_2^2 p_5+\frac{7}{320} p_9) \hbar ^3+O(\hbar ^4)$\\
$(\dfrac{z^3}{3},z)$ & $(\frac{1}{2} p_2 p_1^2+\frac{1}{12}p_4) \hbar +(-\frac{1}{24}p_2^4+\frac{1}{2} p_1^2 p_4 p_2+\frac{1}{24}p_4^2+\frac{1}{3} p_1^3 p_5+\frac{1}{3}p_1 p_7) \hbar ^2
$\\&\qquad $
+(\frac{5}{12} p_8 p_1^4+\frac{2}{3} p_4 p_5 p_1^3+\frac{2}{3} p_2 p_7 p_1^3+\frac{1}{2} p_2 p_4^2 p_1^2+\frac{7}{6} p_{10} p_1^2-\frac{1}{3} p_2^3 p_5 p_1
+\frac{2}{3} p_4 p_7 p_1+\frac{1}{36}p_4^3
$\\&\qquad $
-\frac{1}{6} p_2 p_5^2-\frac{1}{12} p_2^4 p_4-\frac{5}{12} p_2^2 p_8) \hbar ^3+O(\hbar ^4)$\\
$(\dfrac{z^3}{3},\dfrac{z^2}{2})$ & $(\frac{1}{2} p_3 p_1^2+\frac{1}{2} p_2^2 p_1+\frac{5}{24} p_5) \hbar +(\frac{2}{3} p_7 p_1^3+\frac{1}{4} p_4^2 p_1^2+p_3 p_5 p_1^2+\frac{3}{2} p_2 p_6 p_1^2+\frac{1}{6} p_3^3 p_1+p_2^2 p_5 p_1
$\\&\qquad $
+p_2 p_3 p_4 p_1+\frac{9}{4} p_9 p_1+\frac{5}{24} p_5^2+\frac{1}{6} p_2^3 p_4+\frac{3}{8} p_4 p_6+\frac{1}{2}p_3 p_7+\frac{25}{24} p_2 p_8) \hbar ^2+O(\hbar ^3)$\\
\hline
\end{longtable}}

The cases with $s=\pm1$ are well studied. They correspond to the Kontsevich--Witten (for $(r,s)=(2,1)$) and Br{\'e}sin--Gross--Witten (for $(r,s)=(2,-1)$) tau functions and their direct generalizations for the higher $r$'s.  These potentials have natural intersection theory interpretation and, in particular, govern the intersections of Witten and Norbury classes. Moreover, they are described by the Kontsevich-type matrix models and satisfy the KP inspired  W-constraints.  For more details see, e.g., \cite{W,K,MMS,EO-1st,Norbury,AD,chidambaram2023relationsoverlinemathcalmgnnegativerspin}.

%In the case $r=2$ and $s=1$ or $-1$ these are the Kontsevich--Witten and Bresin--Gross--Witten potentials, respectively. Their coefficients have meaning of certain intersection numbers over the moduli spaces of curves.  The cases with higher $r$ and $s=\pm1$ are also well studied and are govern the intersection theory of Witten and Norbury classes.

In the case $r\geq 2$ and $-r<s\leq -1$ the corresponding expansions give intersection numbers with the so-called $\Theta^{(r,s)}$-classes, as it is shown in~\cite{CGS}. An interesting question is whether there is an interpretation for other valued of $s$, for instance $s=0$. 

%Is there a similar treatment of these potentials for other values of~$s$, especially for $s=0$?

If $r\geq 2$ and $-r<s\le1$ then the potential satisfies the $r$-reduction property $\frac{\partial F}{\partial p_{kr}}=0$, $k\ge 1$, except that for $s=0$ these partial derivatives are constants. However, for any $s>1$ the potential is still well defined and nontrivial (even when $s$ is a multiple of~$r$), satisfies KP equations, but does not
show up any reduction property.

Moreover, from the Baker--Akhiezer kernel description \eqref{BAK} it immediately follows that many of these tau functions can be described by the Kontsevich-type matrix models, and the construction \cite[Section 3.5]{ABDKS3} can be applied. Properties of these tau functions can be investigated by the standard matrix models and KP integrability tools. We are not aware of any enumerative meaning of these KP solutions neither for $r>1$ and $s>1$ nor for $r\le1$. 

\subsection{Topological recursion for CohFT with the \texorpdfstring{$\Omega$}{Omega} classes}

Consider the following spectral curve
\begin{equation}\label{eq:Chiodo-TR}
x=\log z-z^r,\quad y=z^s
\end{equation}
for some positive integers $r,s\in\Z_{>0}$. The differentials of the original CEO topological recursion for this spectral curve describe the intersection numbers with the so-called $\Omega$-classes, also called the Chiodo classes in the literature, that compute the total Chern classes of the derived pushforwards of the universal $r$-th roots of the twisted $s$-powers of the log-canonical sheaf~\cite{ABDKS-log-sympl}. 

The point $z=0$ is a unique \nonregular{} point of this spectral curve data which is not a zero of $dx$. We set the point $z=0$ being \yvital{} and all other special points to be \xvital. It follows that the dual differentials $\omega^{\vee,(g)}_n$ are not trivial but they are still much simpler than the differentials of the original topological recursion: they have poles at the origin only and thus they are Laurent polynomial in $z$-variables. We can compute the dual differentials $\omega^{\vee,(g)}_n$ by applying directly the definition of the generalized topological recursion and then obtain the original TR differentials $\omega^{(g)}_n$ by applying $x-y$ swap relations. This leads to the computation of the differentials for the spectral curve~\eqref{eq:Chiodo-TR} that does not require finding explicitly positions of critical points of the~$x$ function, the deck transformation at these points, etc. 
This example is further discussed in~\cite{CGS}. 

%\sasha{A comment of CohFT, something like: We expect that the introduced generalized TR is related to some (to be constructed) natural generalized CohFT in a way that naturally generalized the construction of [DOSS].}

As a concluding remark, it is natural to expect that the generalized topological recursion has the same close connection to the intersection theory on the moduli spaces of curves as the original CEO one~\cite{Eynard-Intersections}, in particular, that is has a connection to some versions of cohomological field theories as in~\cite{DOSS}.

\section{Statements and declarations}

\begin{enumerate}
	\item The authors have no financial or non-financial interests that are directly or indirectly related to the work submitted for publication.
	\item There is no data collected, produced or analyzed within the research related to the paper.
\end{enumerate}

\printbibliography

@misc{ABDKS1,
	title={A universal formula for the $x-y$ swap in topological recursion}, 
	author={Alexander Alexandrov and Boris Bychkov and Petr Dunin-Barkowski and Maxim Kazarian and Sergey Shadrin},
	year={2022},
	eprint={2212.00320},
	archivePrefix={arXiv},
	primaryClass={math-ph},
	addendum   =  {(to appear in J. Eur. Math. Soc.)}
}

@misc{ABDKS3,
	Author = {Alexander Alexandrov and Boris Bychkov and Petr Dunin-Barkowski and Maxim Kazarian and Sergey Shadrin},
	Title = {KP integrability through the $x-y$ swap relation},
	Year = {2023},
	Eprint = {2309.12176},
	Eprinttype = {arXiv},
	primaryClass={math-ph},
	addendum   =  {(to appear in Sel. Math. New Ser.)}
}

@article {ABDKS-logTR-xy,
	AUTHOR = {Alexandrov, A. and Bychkov, B. and Dunin-Barkowski, P. and
	Kazarian, M. and Shadrin, S.},
	title = {Log topological recursion through the prism of x-y swap},
	JOURNAL = {Int. Math. Res. Not. IMRN},
	FJOURNAL = {International Mathematics Research Notices. IMRN},
	YEAR = {2024},
	NUMBER = {21},
	PAGES = {13461--13487},
	ISSN = {1073-7928,1687-0247},
	MRCLASS = {14H81},
	MRNUMBER = {4819863},
	DOI = {10.1093/imrn/rnae213},
	URL = {https://doi.org/10.1093/imrn/rnae213},
	
	xxarchiveprefix = {arXiv},
	xxeprint = {2312.16950},
	xxprimaryclass = {math-ph},
}

@Article{ABDKS-log-sympl,
	Author = {Alexandrov, Alexander and Bychkov, Boris and Dunin-Barkowski, Petr and Kazarian, Maxim and Shadrin, Sergey},
	Title = {Symplectic duality via log topological recursion},
	FJournal = {Communications in Number Theory and Physics},
	Journal = {Commun. Number Theory Phys.},
	ISSN = {1931-4523},
	Volume = {18},
	Number = {4},
	Pages = {795--841},
	Year = {2024},
	Language = {English},
	DOI = {10.4310/CNTP.241203001416},
	zbMATH = {7955389},
	
	xxeprint={2405.10720},
	xxarchivePrefix={arXiv},
	xxprimaryClass={math-ph},
}

@article {alexandrov2023topologicalrecursionsymplecticduality,
	AUTHOR = {Alexandrov, A. and Bychkov, B. and Dunin-Barkowski, P. and
	Kazarian, M. and Shadrin, S.},
	TITLE = {Topological recursion, symplectic duality, and generalized
	fully simple maps},
	JOURNAL = {J. Geom. Phys.},
	FJOURNAL = {Journal of Geometry and Physics},
	VOLUME = {206},
	YEAR = {2024},
	PAGES = {Paper No. 105329, 13},
	ISSN = {0393-0440},
	MRCLASS = {14H81 (30F99 37J06 53D99 81T45)},
	MRNUMBER = {4804150},
	DOI = {10.1016/j.geomphys.2024.105329},
	URL = {https://doi.org/10.1016/j.geomphys.2024.105329},
	
	xxeprint={2304.11687},
	xxarchivePrefix={arXiv},
	xxprimaryClass={math-ph},
}

@misc{alexandrov2024topologicalrecursionrationalspectral,
	title={Any topological recursion on a rational spectral curve is KP integrable}, 
	author={Alexander Alexandrov and Boris Bychkov and Petr Dunin-Barkowski and Maxim Kazarian and Sergey Shadrin},
	year={2024},
	eprint={2406.07391},
	archivePrefix={arXiv},
	primaryClass={math-ph},
	url={https://arxiv.org/abs/2406.07391}, 
}

@article {AD,
	AUTHOR = {Alexandrov, Alexander and Dhara, Saswati},
	TITLE = {On {H}igher {B}r\'ezin--{G}ross--{W}itten {T}au {F}unctions},
	JOURNAL = {Int. Math. Res. Not. IMRN},
	FJOURNAL = {International Mathematics Research Notices. IMRN},
	YEAR = {2025},
	NUMBER = {2},
	PAGES = {rnae286},
	ISSN = {1073-7928,1687-0247},
	MRCLASS = {99-06},
	MRNUMBER = {4852745},
	DOI = {10.1093/imrn/rnae286},
	URL = {https://doi.org/10.1093/imrn/rnae286},
	
	xxarchiveprefix = {arXiv},
	xxeprint = {2204.12273},
	xxprimaryclass = {math-ph},
}

@article {BBCCN,
	AUTHOR = {Borot, Ga\"{e}tan and Bouchard, Vincent and Chidambaram, Nitin K.
	and Creutzig, Thomas and Noshchenko, Dmitry},
	TITLE = {Higher {A}iry structures, {$\mathcal W$} algebras and topological
	recursion},
	JOURNAL = {Mem. Amer. Math. Soc.},
	FJOURNAL = {Memoirs of the American Mathematical Society},
	VOLUME = {296},
	YEAR = {2024},
	NUMBER = {1476},
	PAGES = {v+108},
	ISSN = {0065-9266},
	ISBN = {978-1-4704-6906-1; 978-1-4704-7813-1},
	MRCLASS = {81R10 (14H81)},
	MRNUMBER = {4744795},
	DOI = {10.1090/memo/1476},
	URL = {https://doi.org/10.1090/memo/1476},
}

@misc{limits,
	title={Taking limits in topological recursion},
	author={Gaëtan Borot and Vincent Bouchard and Nitin Kumar Chidambaram and Reinier Kramer and Sergey Shadrin},
	year={2023},
	eprint={2309.01654},
	archivePrefix={arXiv},
	primaryClass={math.AG}
}

@misc{borot2023functional,
	title={Functional relations for higher-order free cumulants},
	author={Gaëtan Borot and Séverin Charbonnier and Elba Garcia-Failde and Felix Leid and Sergey Shadrin},
	year={2023},
	eprint={2112.12184},
	archivePrefix={arXiv},
	primaryClass={math.OA}
}

@article {BEO-loop,
	AUTHOR = {Borot, Ga\"{e}tan and Eynard, Bertrand and Orantin, Nicolas},
	TITLE = {Abstract loop equations, topological recursion and new
	applications},
	JOURNAL = {Commun. Number Theory Phys.},
	FJOURNAL = {Communications in Number Theory and Physics},
	VOLUME = {9},
	YEAR = {2015},
	NUMBER = {1},
	PAGES = {51--187},
	ISSN = {1931-4523},
	MRCLASS = {81T45},
	MRNUMBER = {3339853},
	DOI = {10.4310/CNTP.2015.v9.n1.a2},
	URL = {https://doi.org/10.4310/CNTP.2015.v9.n1.a2},
}

@article {BS-blobbed,
	AUTHOR = {Borot, Ga\"{e}tan and Shadrin, Sergey},
	TITLE = {Blobbed topological recursion: properties and applications},
	JOURNAL = {Math. Proc. Cambridge Philos. Soc.},
	FJOURNAL = {Mathematical Proceedings of the Cambridge Philosophical
	Society},
	VOLUME = {162},
	YEAR = {2017},
	NUMBER = {1},
	PAGES = {39--87},
	ISSN = {0305-0041},
	MRCLASS = {14H81 (32G99 81T30)},
	MRNUMBER = {3581899},
	MRREVIEWER = {Xiaobin Li},
	DOI = {10.1017/S0305004116000323},
	URL = {https://doi.org/10.1017/S0305004116000323},
}

@article {BE,
    AUTHOR = {Bouchard, Vincent and Eynard, Bertrand},
     TITLE = {Think globally, compute locally},
   JOURNAL = {J. High Energy Phys.},
  FJOURNAL = {Journal of High Energy Physics},
      YEAR = {2013},
    NUMBER = {2},
     PAGES = {143, front matter + 34},
      ISSN = {1126-6708,1029-8479},
   MRCLASS = {81T45 (30F10 32G15)},
  MRNUMBER = {3046532},
MRREVIEWER = {Lee-Peng\ Teo},
       DOI = {10.1007/JHEP02(2013)143},
       URL = {https://doi.org/10.1007/JHEP02(2013)143},
}

@Article{BDKS1,
	Author = {Bychkov, Boris and Dunin-Barkowski, Petr and Kazarian, Maxim and Shadrin, Sergey},
	Title = {Explicit closed algebraic formulas for {Orlov}-{Scherbin} {{\(n\)}}-point functions},
	FJournal = {Journal de l'{\'E}cole Polytechnique -- Math{\'e}matiques},
	Journal = {J. {\'E}c. Polytech., Math.},
	ISSN = {2429-7100},
	Volume = {9},
	Pages = {1121--1158},
	Year = {2022},
	DOI = {10.5802/jep.202},
	zbMATH = {7559604},
	Zbl = {1504.37080}
}

@article{BDKS2,
	author = {Bychkov, Boris and Dunin-Barkowski, Petr and Kazarian, Maxim and Shadrin, Sergey},
	title = {Topological recursion for Kadomtsev–Petviashvili tau functions of hypergeometric type},
	journal = {J. Lond. Math. Soc.},
	volume = {109},
	number = {6},
	pages = {e12946},
	doi = {10.1112/jlms.12946},
	url = {https://doi.org/10.1112/jlms.12946},
	year = {2024}
}

@article {BDKS3,
	AUTHOR = {Bychkov, Boris and Dunin-Barkowski, Petr and Kazarian, Maxim
	and Shadrin, Sergey},
	TITLE = {Generalised ordinary vs fully simple duality for {$n$}-point
	functions and a proof of the {B}orot-{G}arcia-{F}ailde
	conjecture},
	JOURNAL = {Comm. Math. Phys.},
	FJOURNAL = {Communications in Mathematical Physics},
	VOLUME = {402},
	YEAR = {2023},
	NUMBER = {1},
	PAGES = {665--694},
	ISSN = {0010-3616,1432-0916},
	MRCLASS = {14H81 (46N50 57K20)},
	MRNUMBER = {4616685},
	MRREVIEWER = {Vida\ Milani},
	DOI = {10.1007/s00220-023-04732-7},
	URL = {https://doi.org/10.1007/s00220-023-04732-7},
}

@article {bychkov2023symplecticdualitytopologicalrecursion,
	AUTHOR = {Bychkov, Boris and Dunin-Barkowski, Petr and Kazarian, Maxim
	and Shadrin, Sergey},
	TITLE = {Symplectic duality for topological recursion},
	JOURNAL = {Trans. Amer. Math. Soc.},
	FJOURNAL = {Transactions of the American Mathematical Society},
	VOLUME = {378},
	YEAR = {2025},
	NUMBER = {2},
	PAGES = {1001--1054},
	ISSN = {0002-9947},
	MRCLASS = {05A15 (05E14 14N10 37K10 81T45)},
	MRNUMBER = {4850433},
	DOI = {10.1090/tran/9352},
	URL = {https://doi.org/10.1090/tran/9352},
	
	xxeprint={2206.14792},
	xxarchivePrefix={arXiv},
	xxprimaryClass={math-ph},
}

@article {BDS-BMSnumbers,
	AUTHOR = {Bychkov, B. and Dunin-Barkowski, P. and Shadrin, S.},
	TITLE = {Combinatorics of {B}ousquet-{M}\'{e}lou-{S}chaeffer numbers in the
	light of topological recursion},
	JOURNAL = {European J. Combin.},
	FJOURNAL = {European Journal of Combinatorics},
	VOLUME = {90},
	YEAR = {2020},
	PAGES = {103184, 35},
	ISSN = {0195-6698},
	MRCLASS = {05E14 (05A05 05A15 14H81)},
	MRNUMBER = {4125529},
	DOI = {10.1016/j.ejc.2020.103184},
	URL = {https://doi.org/10.1016/j.ejc.2020.103184},
}

@article {CEO,
	AUTHOR = {Chekhov, Leonid and Eynard, Bertrand and Orantin, Nicolas},
	TITLE = {Free energy topological expansion for the 2-matrix model},
	JOURNAL = {J. High Energy Phys.},
	FJOURNAL = {Journal of High Energy Physics. A SISSA Journal},
	YEAR = {2006},
	NUMBER = {12},
	PAGES = {053, 31},
	ISSN = {1126-6708},
	MRCLASS = {81T45 (15A52)},
	MRNUMBER = {2276699},
	MRREVIEWER = {Gernot Akemann},
	DOI = {10.1088/1126-6708/2006/12/053},
	URL = {https://doi.org/10.1088/1126-6708/2006/12/053},
}

@article {CN,
    AUTHOR = {Chekhov, Leonid and Norbury, Paul},
     TITLE = {Topological recursion with hard edges},
   JOURNAL = {Internat. J. Math.},
  FJOURNAL = {International Journal of Mathematics},
    VOLUME = {30},
      YEAR = {2019},
    NUMBER = {3},
     PAGES = {1950014, 29},
      ISSN = {0129-167X,1793-6519},
   MRCLASS = {14N10 (14H81 32G15)},
  MRNUMBER = {3941980},
MRREVIEWER = {Piotr\ Su\l kowski},
       DOI = {10.1142/S0129167X19500149},
       URL = {https://doi.org/10.1142/S0129167X19500149},
}

@misc{chidambaram2023relationsoverlinemathcalmgnnegativerspin,
	title={Relations on $\overline{\mathcal{M}}_{g,n}$ and the negative $r$-spin Witten conjecture}, 
	author={Nitin Kumar Chidambaram and Elba Garcia-Failde and Alessandro Giacchetto},
	year={2023},
	eprint={2205.15621},
	archivePrefix={arXiv},
	primaryClass={math.AG},
	url={https://arxiv.org/abs/2205.15621}, 
}

@misc{CGS,
	title={A paper in preparation.},
	author={Vincent Bouchard and Nitin Kumar Chidambaram and Alessandro Giacchetto and Sergey Shadrin},
	year={2025},
}

@article {DOSS,
	AUTHOR = {Dunin-Barkowski, P. and Orantin, N. and Shadrin, S. and Spitz,
	L.},
	TITLE = {Identification of the {G}ivental formula with the spectral
	curve topological recursion procedure},
	JOURNAL = {Comm. Math. Phys.},
	FJOURNAL = {Communications in Mathematical Physics},
	VOLUME = {328},
	YEAR = {2014},
	NUMBER = {2},
	PAGES = {669--700},
	ISSN = {0010-3616},
	MRCLASS = {81T45 (14N35 53D45)},
	MRNUMBER = {3199996},
	MRREVIEWER = {Wan Keng Cheong},
	DOI = {10.1007/s00220-014-1887-2},
	URL = {https://doi.org/10.1007/s00220-014-1887-2},
}

@article {Eynard-Intersections,
	AUTHOR = {Eynard, B.},
	TITLE = {Invariants of spectral curves and intersection theory of
	moduli spaces of complex curves},
	JOURNAL = {Commun. Number Theory Phys.},
	FJOURNAL = {Communications in Number Theory and Physics},
	VOLUME = {8},
	YEAR = {2014},
	NUMBER = {3},
	PAGES = {541--588},
	ISSN = {1931-4523},
	MRCLASS = {14H10 (11G05 14C17 32G15)},
	MRNUMBER = {3282995},
	MRREVIEWER = {Letterio Gatto},
	DOI = {10.4310/CNTP.2014.v8.n3.a4},
	URL = {https://doi.org/10.4310/CNTP.2014.v8.n3.a4},
}

@article {EO-1st,
	AUTHOR = {Eynard, B. and Orantin, N.},
	TITLE = {Invariants of algebraic curves and topological expansion},
	JOURNAL = {Commun. Number Theory Phys.},
	FJOURNAL = {Communications in Number Theory and Physics},
	VOLUME = {1},
	YEAR = {2007},
	NUMBER = {2},
	PAGES = {347--452},
	ISSN = {1931-4523},
	MRCLASS = {14H15 (14N35 32A27 37K10 37K20 81T45)},
	MRNUMBER = {2346575},
	MRREVIEWER = {Vincent Bouchard},
	DOI = {10.4310/CNTP.2007.v1.n2.a4},
	URL = {https://doi.org/10.4310/CNTP.2007.v1.n2.a4},
}

@article {EO-xy,
	AUTHOR = {Eynard, B. and Orantin, N.},
	TITLE = {Topological expansion of mixed correlations in the {H}ermitian
	2-matrix model and {$x$}-{$y$} symmetry of the {$F_g$}
	algebraic invariants},
	JOURNAL = {J. Phys. A},
	FJOURNAL = {Journal of Physics. A. Mathematical and Theoretical},
	VOLUME = {41},
	YEAR = {2008},
	NUMBER = {1},
	PAGES = {015203, 28},
	ISSN = {1751-8113,1751-8121},
	MRCLASS = {14H15 (14N10 14N35 82B41)},
	MRNUMBER = {2450700},
	MRREVIEWER = {Brad\ Safnuk},
	DOI = {10.1088/1751-8113/41/1/015203},
	URL = {https://doi.org/10.1088/1751-8113/41/1/015203},
}

@article {hock2022xy,
	AUTHOR = {Hock, Alexander},
	TITLE = {On the {$x$}-{$y$} {S}ymmetry of {C}orrelators in
	{T}opological {R}ecursion via {L}oop {I}nsertion {O}perator},
	JOURNAL = {Comm. Math. Phys.},
	FJOURNAL = {Communications in Mathematical Physics},
	VOLUME = {405},
	YEAR = {2024},
	NUMBER = {7},
	PAGES = {Paper No. 166},
	ISSN = {0010-3616,1432-0916},
	MRCLASS = {46L54 (05 15B57)},
	MRNUMBER = {4768536},
	DOI = {10.1007/s00220-024-05043-1},
	URL = {https://doi.org/10.1007/s00220-024-05043-1},
}

@article {hock2023laplace,
	AUTHOR = {Hock, Alexander},
	TITLE = {Laplace transform of the {$x - y$} symplectic transformation
	formula in topological recursion},
	JOURNAL = {Commun. Number Theory Phys.},
	FJOURNAL = {Communications in Number Theory and Physics},
	VOLUME = {17},
	YEAR = {2023},
	NUMBER = {4},
	PAGES = {821--845},
	ISSN = {1931-4523,1931-4531},
	MRCLASS = {14N10 (05A15 14H70 14H81 30F30)},
	MRNUMBER = {4704940},
	DOI = {10.4310/CNTP.2023.v17.n4.a1},
}

@article {hock2022simple,
	AUTHOR = {Hock, Alexander},
	TITLE = {A simple formula for the {$x$}-{$y$} symplectic transformation
	in topological recursion},
	JOURNAL = {J. Geom. Phys.},
	FJOURNAL = {Journal of Geometry and Physics},
	VOLUME = {194},
	YEAR = {2023},
	PAGES = {Paper No. 105027, 26},
	ISSN = {0393-0440,1879-1662},
	MRCLASS = {46L54 (15B52 16R60)},
	MRNUMBER = {4659813},
	DOI = {10.1016/j.geomphys.2023.105027},
	URL = {https://doi.org/10.1016/j.geomphys.2023.105027},
}

@article {hock2023xy,
	AUTHOR = {Hock, Alexander},
	TITLE = {{$x - y$} duality in topological recursion for exponential
	variables via quantum dilogarithm},
	JOURNAL = {SciPost Phys.},
	FJOURNAL = {SciPost Physics},
	VOLUME = {17},
	YEAR = {2024},
	NUMBER = {2},
	PAGES = {Paper No. 065, 36},
	ISSN = {2542-4653},
	MRCLASS = {14H81 (05E14 32G34 33B15 81T30)},
	MRNUMBER = {4794300},
	DOI = {10.21468/SciPostPhys.17.2.065},
}

@article {Kharchev-Marshakov,
	AUTHOR = {Kharchev, S. and Marshakov, A.},
	TITLE = {On {$p$}-{$q$} duality and explicit solutions in {$c\leq 1$}
	{$2$}D gravity models},
	JOURNAL = {Internat. J. Modern Phys. A},
	FJOURNAL = {International Journal of Modern Physics A. Particles and
	Fields. Gravitation. Cosmology},
	VOLUME = {10},
	YEAR = {1995},
	NUMBER = {8},
	PAGES = {1219--1236},
	ISSN = {0217-751X,1793-656X},
	MRCLASS = {81T30 (81T40)},
	MRNUMBER = {1321929},
	MRREVIEWER = {Igor\ Polyubin},
	DOI = {10.1142/S0217751X95000577},
	URL = {https://doi.org/10.1142/S0217751X95000577},

	xxeprint={hep-th/9303100},
	xxarchivePrefix={arXiv},
}

@article {K,
    AUTHOR = {Kontsevich, Maxim},
     TITLE = {Intersection theory on the moduli space of curves and the
              matrix {A}iry function},
   JOURNAL = {Comm. Math. Phys.},
  FJOURNAL = {Communications in Mathematical Physics},
    VOLUME = {147},
      YEAR = {1992},
    NUMBER = {1},
     PAGES = {1--23},
      ISSN = {0010-3616,1432-0916},
   MRCLASS = {32G15 (14H15 58F07 81T40)},
  MRNUMBER = {1171758},
MRREVIEWER = {Claude\ Itzykson},
       URL = {http://projecteuclid.org/euclid.cmp/1104250524},
       DOI = {10.1007/BF02099526},
}

@article {LPSZ,
	AUTHOR = {Lewanski, Danilo and Popolitov, Alexandr and Shadrin, Sergey
	and Zvonkine, Dimitri},
	TITLE = {Chiodo formulas for the {$r$}-th roots and topological
	recursion},
	JOURNAL = {Lett. Math. Phys.},
	FJOURNAL = {Letters in Mathematical Physics},
	VOLUME = {107},
	YEAR = {2017},
	NUMBER = {5},
	PAGES = {901--919},
	ISSN = {0377-9017},
	MRCLASS = {14H10 (14N10 14N35)},
	MRNUMBER = {3633029},
	MRREVIEWER = {Fabio Perroni},
	DOI = {10.1007/s11005-016-0928-5},
	URL = {https://doi.org/10.1007/s11005-016-0928-5},
}

@article {MMS,
    AUTHOR = {Mironov, A. and Morozov, A. and Semenoff, G. W.},
     TITLE = {Unitary matrix integrals in the framework of the generalized
              {K}ontsevich model},
   JOURNAL = {Internat. J. Modern Phys. A},
  FJOURNAL = {International Journal of Modern Physics A. Particles and
              Fields. Gravitation. Cosmology},
    VOLUME = {11},
      YEAR = {1996},
    NUMBER = {28},
     PAGES = {5031--5080},
      ISSN = {0217-751X,1793-656X},
   MRCLASS = {81T40 (58C35 58F07)},
  MRNUMBER = {1411777},
MRREVIEWER = {Igor\ Polyubin},
       DOI = {10.1142/S0217751X96002339},
       URL = {https://doi.org/10.1142/S0217751X96002339},
}

@article {Norbury,
	AUTHOR = {Norbury, Paul},
	TITLE = {A new cohomology class on the moduli space of curves},
	JOURNAL = {Geom. Topol.},
	FJOURNAL = {Geometry \& Topology},
	VOLUME = {27},
	YEAR = {2023},
	NUMBER = {7},
	PAGES = {2695--2761},
	ISSN = {1465-3060,1364-0380},
	MRCLASS = {14D23 (32G15 53D45)},
	MRNUMBER = {4645485},
	DOI = {10.2140/gt.2023.27.2695},
	URL = {https://doi.org/10.2140/gt.2023.27.2695},
}

@incollection {W,
    AUTHOR = {Witten, Edward},
     TITLE = {Two-dimensional gravity and intersection theory on moduli
              space},
 BOOKTITLE = {Surveys in differential geometry ({C}ambridge, {MA}, 1990)},
     PAGES = {243--310},
 PUBLISHER = {Lehigh Univ., Bethlehem, PA},
      YEAR = {1991},
      ISBN = {0-8218-0168-6},
   MRCLASS = {32G15 (14C17 14H15 32G81 58F07 81T40)},
  MRNUMBER = {1144529},
MRREVIEWER = {Steven\ Rosenberg},
}

\end{document}